\documentclass[11pt,a4paper]{article}

\usepackage[english]{babel}

\usepackage[letterpaper,top=2cm,bottom=2cm,left=3cm,right=3cm,marginparwidth=1.75cm]{geometry}

\usepackage[shortlabels]{enumitem}

\usepackage{amsfonts}
\usepackage{amsmath}
\usepackage{amsthm}
\usepackage{amssymb}
\usepackage{graphicx}
\usepackage{microtype}
\usepackage{xspace}
\usepackage{xparse} 
\usepackage[dvipsnames]{xcolor}
\usepackage[colorlinks=true, allcolors=teal]{hyperref}
\usepackage[capitalise,nameinlink,noabbrev]{cleveref}

\usepackage{pgfplots}
\pgfplotsset{compat=1.17}
\usepgfplotslibrary{groupplots}

\usepackage{mathtools}
\usepackage{leftindex}

\theoremstyle{definition}

\newtheorem{theorem}{Theorem}
\newtheorem{lemma}[theorem]{Lemma}
\newtheorem{corollary}[theorem]{Corollary}
\newtheorem{proposition}[theorem]{Proposition}

\theoremstyle{remark}

\newtheorem{example}{Example}

\usepackage[T1]{fontenc}

\usepackage{algorithm}
\usepackage[noend]{algorithmic}
\floatname{algorithm}{Algorithm}

\newcommand{\preals}{\ensuremath{\mathbb{R}_{\geq 0}}\xspace}
\newcommand{\naturals}{\ensuremath{\mathbb{N}\xspace}}
\newcommand{\tvd}{\ensuremath{\text{TV}}\xspace}
\newcommand{\acc}{\ensuremath{\eta}\xspace}
\newcommand{\tmax}{\ensuremath{T_{\max}}\xspace}

\newcommand{\aefx}{\ensuremath{a\text{-EFX}}\xspace}
\newcommand{\praefx}{\ensuremath{a'\text{-EFX}}\xspace}
\newcommand{\aefo}{\ensuremath{a\text{-EF1}}\xspace}

\newcommand{\tila}{\ensuremath{\Tilde{a}}\xspace}

\providecommand{\ceil}[1]{\ensuremath{\left \lceil #1 \right \rceil }}
\providecommand{\floor}[1]{\ensuremath{\left \lfloor #1 \right \rfloor }}
\providecommand{\TV}[2]{\ensuremath{\tvd \left( #1, #2 \right) }}

\usepackage{cancel}

\newcommand{\oset}[1]{\ensuremath{\leftindex_{1} {#1}}}
\newcommand{\dm}{\ensuremath{d_{\max}}}
\newcommand{\dmax}{\ensuremath{\frac{(4 + a - a^2)(1 - a)}{(2 + a)(5 - a)(1 + a)}}}
\newcommand{\dmaxformone}{\ensuremath{\frac{2}{3} \cdot \frac{1 - a}{2 + a}}}
\newcommand{\dmaxbin}{\ensuremath{\frac{2}{5}\cdot \frac{1 - a}{1 + a}}}
\newcommand{\dmaxthree}{\ensuremath{\frac{1 - a}{1+a}}}
\newcommand{\pUB}{\ensuremath{\frac{1+a}{1-a} \cdot \dm}}
\newcommand{\pAoneUB}{\ensuremath{\frac{4 + a - a^2}{(2 + a)(5 - a)}}}
\newcommand{\thres}{\ensuremath{\frac{(1 - a)^2}{(2 + a)(5 - a)}}}
\newcommand{\incr}{\ensuremath{\frac{(3 + a)(1 - a)}{(2 + a)(5 - a)(1 + a)}}}

\newcommand{\xset}[1]{\ensuremath{\leftindex_{\text{x}} {#1}}}

\title{Online EFX Allocations with Predictions\\ \text{}}

\author{Themistoklis Melissourgos\\
        \small{University of Essex, UK}\\
		\small{\texttt{themistoklis.melissourgos@essex.ac.uk}}
        \and
        Nicos Protopapas\\
        \small{Archimedes, Athena Research Center, Greece}\\
        \small{\texttt{n.protopapas@athenarc.gr}}
        }
        
\date{\vspace{0.5cm}}

\begin{document}
\maketitle

\begin{abstract}
We study an online fair division problem where a fixed number of goods arrive sequentially and must be allocated to a given set of agents. Once a good arrives, its true value for each agent is revealed, and it has to be immediately and irrevocably allocated to some agent. The ultimate goal is to ensure \emph{envy-freeness up to any good (EFX)} after all goods have been allocated. Unfortunately, as we show, approximate EFX allocations are unattainable in general, even under restrictive assumptions on the valuation functions.

To address this, we follow a recent and fruitful trend of \emph{augmenting algorithms with predictions}. Specifically, we assume access to a prediction vector estimating the agents’ true valuations -- e.g., generated by a machine learning model trained on past data. Predictions may be unreliable, and we measure their error using the total variation distance from the true valuations, that is, the percentage of predicted value-mass that disagrees with the true values.

Focusing on the natural class of additive valuations, we prove impossibility results even on approximate EFX allocations for algorithms that either ignore predictions or rely solely on them. We then turn to algorithms that use both the predictions and the true values and show strong lower bounds on the prediction accuracy that is required by any algorithm to compute an approximate EFX. These negative results persist even under identical valuations, contrary to the offline setting where exact EFX allocations always exist without the necessity of predictions. We then present an algorithm for two agents with identical valuations that uses effectively the predictions and the true values. The algorithm approximates EFX, with its guarantees improving as the accuracy of the predictions increases.
\end{abstract}


\section{Introduction}

\emph{Envy-freeness (EF)} is one of the most natural and well-studied fair division notions. In the typical setting, there is a common set of resources that needs to be allocated among a given set of agents, each having a personal valuation function for the resources. In the case of divisible resources, the agents can be allocated an arbitrary portion of each resource, and it is known that EF allocations always exist even when restrictions are imposed on the solution~\cite{Stromquist80}. On the other hand, when only \emph{indivisible resources} are available, each of them has to be allocated intact to some agent, and it is easy to show that even approximate EF allocations are not guaranteed to exist for any positive approximation factor: consider the simple case where a single good has to be allocated between two agents that value it more than getting nothing. 

This non-existence has sparked a fruitful line of work that, over the last 15 years, has proposed several envy-freeness relaxations as close to exact envy-freeness as possible, aiming to prove they always exist. The most prominent of those relaxations is the notion of \emph{envy-freeness up to any good (EFX)}, which allows agents to envy others as long as the envy is eliminated with the virtual removal of a good with the lowest marginal value from the envied agent's bundle. A more relaxed notion is that of \emph{envy-freeness up to one good (EF1)}, where envy can be eliminated by the virtual removal of a good with the highest marginal value from the envied agent's bundle. The EF1 notion was formally defined by~\cite{budish2011combinatorial} and~\cite{lipton2004approximately} under a different name, and it 
 was proven in the same works that it always exists for increasingly broader classes of valuation functions. The stronger notion of EFX was defined independently by~\cite{caragiannis2019unreasonable} and~\cite{gourves2014near} and besides specific special cases, most notably for up to three agents~\cite{CGM24}, a proof of existence is elusive. In particular, even though it is widely conjectured that EFX allocations always exist for any number of agents in the broad class of monotone valuation functions, this has not been refuted so far, and it remains open whether EFX allocations exist even for four agents with additive valuations. Further results are discussed in the excellent survey of~\cite{amanatidis2023fair}.

To amend the efforts toward a proof of EFX existence, the notion of \emph{approximate envy-freeness up to any good (\aefx)} was proposed in~\cite{plaut2020almost}. An allocation is \aefx for some given $a \in [0,1]$ if, for every agent, her allocated bundle's valuation is no smaller than $a$ times that of any other agent's bundle when the least-valued good is removed from the latter bundle. It is known from~\cite{amanatidis2020multiple}  that $(\varphi - 1)$-EFX allocations always exist for any number $n$ of agents with additive valuations. This approximation factor was recently improved to $2/3$ for the case $n \leq 7$ in~\cite{AFS24}. 

Beyond classical offline settings, some real-life scenarios involve goods arriving over time and requiring immediate, irrevocable allocation to agents, with fairness guaranteed only after the final good is allocated. If all values were known in advance, such problems could be reduced to the offline case; but without this information, computing a fair allocation becomes significantly more challenging.

A typical example of the \emph{online} setting would be the following. A foodbank collects surplus or near-expiry food \emph{packets} daily from supermarkets and charities, and distributes them to \emph{endpoints} like community kitchens, schools, or shelters. Each packet contains a mix of products from a source, but its exact contents are unknown to the foodbank and the endpoints prior to its arrival. Daily, multiple packets arrive asynchronously at the foodbank, but for ease of presentation, let us consider a single packet arriving each day. The packets are also perishable, so they must be inspected, categorized, and \emph{assigned to an endpoint immediately}, i.e., the same day. The common practice currently among foodbanks is that allocation decisions consider mostly the urgency, relevance, proximity and capacity of endpoints. This fast and complex decision-making can cause significant envy between endpoints, especially as foodbanks typically do not share data, requiring endpoints to trust the process blindly~\cite{AKKERMAN2023108926}.

To address this, we propose a transparent online fair division model where all packet contents and assignments are made visible to all endpoints.\footnote{This can be done anonymously to respect privacy.} The central question now becomes: \emph{Can fairness be ensured in such a transparent system?} We aim to design algorithms that, over a specified time period, produce \aefx allocations for the highest possible value of $a \in [0,1]$. In our model, endpoints are agents with personal valuations. At the start of the month, the foodbank gives all endpoints an estimate of the total contents expected over 30 days. As each packet arrives and is inspected, each endpoint reports its \emph{relative value} as a fraction of the expected month's total, and the foodbank allocates it based on these values.

This simplified, transparent setup removes much of the aforementioned decision complexity, as the only input required from endpoints now is the current relative valuation. However, a key challenge remains: even if the total set of products is known, their distribution across packets is not. We show that this uncertainty prevents any algorithm from guaranteeing \aefx allocations for any positive $a$. To help the algorithm of the foodbank achieve approximate EFX allocations, we allow \emph{predictions} in the model, following the recently introduced paradigm of \emph{algorithmic design with predictions}~\cite{LykourisV21}. At the start of the month, the foodbank provides estimated contents for each future packet (e.g., derived via machine learning). Endpoints respond with their estimated relative value for the packets, i.e., a vector of predicted values. Then, as each packet arrives, the foodbank reveals its true contents to the endpoints and asks for their true relative value (which might differ from the respective prediction). Finally, the foodbank runs an \aefx algorithm which prescribes who will receive the current packet. The foodbank has a contract with its endpoints, guaranteeing that on day 30, the allocation of packets that have been arriving since day 1 will be an \aefx allocation for some promised value of $a \in [0,1]$. 

We focus on additive valuation functions, one for each agent $i \in N$, over a finite set of goods (one for each time-step). The metric we use to capture the distance between the prediction vector $p_i$ and the vector of true values $v_i$ is the \emph{total variation distance (\tvd distance)}. In order to have a meaningful, uniform magnitude of distance, we consider normalized valuations, where the empty set has value 0 and the set of all goods has value 1 to all agents. Therefore, $p_i$ and $v_i$ can be thought of as probability distributions, and their \tvd distance (also called \emph{error}) $d_i$ captures the total unsuccessfully predicted value as a fraction of the entire value that will arrive over time. The quantity $\eta_i = 1- d_i$ is the prediction accuracy of agent $i$, and we consider also the worst guarantees over agents, that is, $D := \max_{i \in N} d_i$, and $\eta := 1 - D$. The goods arrive over time, and their true value appears to the agents only upon their arrival. An algorithm takes as input the predictions $(p_i)_{i \in N}$ and the accuracy levels $(\eta_i)_{i \in N}$ and at each time-step $t = 1,\dots,T$, the true valuation $v_{i}(g_t)$ of agent $i$ for good $g_t$ is revealed to the algorithm; then the algorithm has to irrevocably allocate the good to an agent. The objective is, for a given level of accuracy $\eta$, to design an algorithm that computes an \aefx with the maximum possible $a \in [0,1]$.

\subsection{Our results}

We pose the question: \emph{How does the quality (approximation factor $a$) of \aefx allocations depend on the prediction accuracy of the agents?} We consider additive valuations and provide bounds on the prediction accuracy (or equivalently, on the allowed prediction error) as a function of $a \in [0,1]$.
In \cref{sec:alg_without_predictions}, we explore the limitations of algorithms that do not have per-item predictions (but know the value of the whole set of goods for each agent, i.e., their valuations are normalized). Our main results show that for two agents with identical valuations, one can achieve $(\varphi - 1)$-EFX using a simple threshold-based algorithm, while no \aefx algorithm without predictions exists for any $a \in (\varphi - 1, 1]$. When the valuations are not restricted to be identical, the latter impossibility result holds for any $a \in (0, 1]$. Therefore, to achieve any improvement on $a$ we turn our focus to algorithms with predictions in \cref{sec:alg_with_predictions}. 

We start in \cref{sec:alg_with_predictions_entirely} by exploring the limitations of the other extreme, i.e., algorithms that rely entirely on predictions, disregarding the true values. We show that for $n \geq 2$ agents with accuracy at least $1 - \frac{\tila - a}{(2n-2+\tila)(1+a)}$ for some given $a, \tila \in \left[ 0, 1 \right]$ with $a \leq \tila$, if an algorithm can compute an $\tila$-EFX allocation on the predicted values, then a slight modification of it can compute an \aefx allocation on the true values. We show that when $\tila = 1$, this accuracy is also necessary (the bound is tight) among algorithms oblivious to the true values. Next, in \cref{sec:alg_with_predictions_and_true}, we study the limitations of algorithms that use both predictions and true valuations, and show lower bounds on the level of accuracy as a function of the desired $a$. In particular, we show that for 2 agents, accuracy of $1 - \frac{1-a}{\min\{ 6a, 4 \}}$ is needed by any \aefx algorithm for $a \in \left( \frac{1}{2}, 1 \right]$. This bound slightly improves when the agents have identical valuations, where the necessary accuracy of an \aefx algorithm becomes $1 - \frac{1-a}{\min\{2a(2+a), 4\}}$ for any $a \in (\varphi - 1, 1]$. We show similarly strong bounds for $n \geq 3$ agents with identical valuations, even for $a \in (0, 1]$. All our impossibility results are proven using $k$-value predictions (i.e., vectors with only $k \in \naturals$ distinct values) for small $k$, making them particularly strong. 

Finally, we attempt to bridge the gap between the lower $\left( 1 - \frac{1-a}{\min\{2a(2+a), 4\}} \right)$ and the upper bound $\left( 1 - \frac{1 - a}{3(1+a)} \right)$ on the accuracy for 2 agents with identical valuations. We manage to do so with \cref{alg:id_2_ag} which guarantees an \aefx allocation for any $a \in (\varphi - 1, 1]$ while using predictions of accuracy $1 - \dmax$. Furthermore, we use this algorithm to derive an improved upper bound of $1 - \frac{2}{5} \frac{1-a}{1+a}$ when predictions are 2-value functions, and complement it with a lower bound of $1 - \frac{1-a}{2}$.

\cref{fig:2_agents_identical} illustrates our most important results, namely the bounds derived on the maximum prediction error $D$ as a function of $a \in [0, 1]$ for the case of 2 agents with additive, identical valuations. The red solid plot visualizes the error sufficient for \cref{alg:id_2_ag} (our main positive result, \cref{thm:with_predictions_2_ag_ternary_a-EFX_positive}) to derive an \aefx  allocation for $a \in (\varphi-1, 1]$ (which becomes 1 for $a \in [0,\varphi-1]$ due to \cref{thm:no_predictions_phi-EFX_upper_bound}). The blue dash-dot plot shows the (tight) bound on the error of algorithms without predictions (\cref{cor:with_predictions_n_ag_a-EFX_positive}), and the green dashed plot shows the highest error that any algorithm with predictions can have, due to \cref{thm:with_predictions_comb-beyond-phi-EFX_inapprox_bound}.

 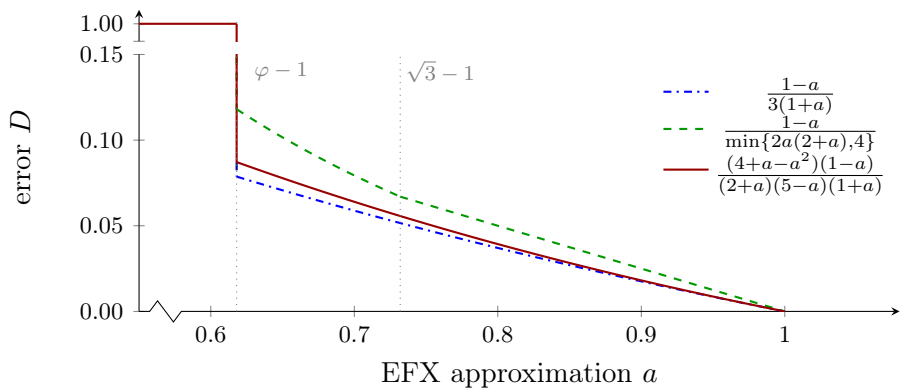
\begin{figure}[H]
     \centering
     \begin{tikzpicture}
\begin{groupplot}[
  group style={
    group name=broken,
    group size=1 by 2,
    vertical sep=5pt
  },
  tick label style={font=\footnotesize},
  width=0.75*\textwidth,
  xmin=0.55, xmax=1.08,
  xtick={0.6, 0.7, 0.8, 0.9, 1},
  xticklabels={0.6, 0.7, 0.8, 0.9, 1},
  axis x discontinuity=crunch,
  domain=0.6:1.1,
  samples=200,
  scaled y ticks=false,
  scaled x ticks=false,
  tick label style={
    /pgf/number format/fixed,
    /pgf/number format/precision=2,
    /pgf/number format/fixed zerofill
  }
]

\nextgroupplot[
  height=2cm,
  ymin=0.9, ymax=1.08,
  ytick={0.9,1},
  yticklabels={,1.00},
  xticklabels={},
  y axis line style={-}, 
  axis lines=left,
  axis x line=none,
]

\addplot[red!60!black, thick, domain=0:0.618] {1};

\addplot[red!60!black, thick] coordinates {(0.618, 1.0) (0.618, 0.118)};

\addplot[dotted, gray] coordinates {(0.618, 1.01) (0.618, 0.98)};

\nextgroupplot[
  height=5cm,
  ymin=0, ymax=0.15,
  axis lines=left,
  y axis line style={-}, 
  ytick={0, 0.05, 0.1, 0.15},
  xlabel={EFX approximation $a$},
  ylabel style={at={(axis description cs:-0.13,0.6)}},
  ylabel={error $D$},
  legend style={
    at={(1,0.95)},
    anchor=north east,
    draw=none,
    fill=none,
    font=\small
  }
]

\addlegendimage{line legend, blue, thick, dash dot}
\addlegendentry{$\frac{1 - a}{3(1 + a)}$}

\addlegendimage{line legend, green!60!black, thick, dashed}
\addlegendentry{$\frac{1 - a}{\min\{2a(2 + a), 4\}}$}

\addlegendimage{line legend, red!60!black, thick}
\addlegendentry{$\frac{(4 + a - a^2)(1-a)}{(2 + a)(5-a)(1+a)}$}


\addplot[green!60!black, thick, dashed] coordinates {(0.618, 1.0) (0.618, 0.118)};

\addplot[red!60!black, thick] coordinates {(0.618, 1) (0.618, 0.087)};

\addplot[blue!60!black, dash dot, thick] coordinates {(0.618, 0.087) (0.618, 0.078)};

\addplot[blue, thick, dash dot, domain=0.618:1] 
  {(1 - x)/(3 * (1 + x))};

\addplot[green!60!black, thick, dashed, domain=0.618:0.732] 
  {(1 - x)/(2 * x * (2 + x))};
\addplot[green!60!black, dashed, thick, domain=0.732:1] 
  {(1 - x)/4};

\addplot[red!60!black, thick, domain=0.618:1]
  {((4+x-x*x)*(1-x))/((2+x)*(5-x)*(1+x))};

\addplot[dotted, gray] coordinates {(0.618, 0) (0.618, 0.15)};
\node[gray] at (axis cs:0.649,0.14) {\scriptsize$\varphi - 1$};

\addplot[dotted, gray] coordinates {(0.732, 0) (0.732, 0.15)};
\node[gray] at (axis cs:0.76,0.14) {\scriptsize$\sqrt{3} - 1$};

\end{groupplot}
\end{tikzpicture}    
     \caption{Prediction error $D=1-\eta$ as a function of $a \in [0, 1]$ for 2 agents with identical valuations. Dash-dot plot: \cref{cor:with_predictions_n_ag_a-EFX_positive}, dashed plot: \cref{thm:with_predictions_comb-beyond-phi-EFX_inapprox_bound}, solid plot: \cref{thm:no_predictions_phi-EFX_upper_bound} and \cref{thm:with_predictions_2_ag_ternary_a-EFX_positive}.}
     \label{fig:2_agents_identical}
 \end{figure}

\subsection{Further related work}

Early work on online fair division was initiated by~\cite{AleksandrovAGW15}, who analyzed two simple randomized mechanisms with respect to ex-ante and ex-post envy-freeness. One of the mechanisms was shown to satisfy ex-ante envy-freeness, but under general additive utilities, neither could guarantee any bounded form of ex-post envy-freeness. A broader survey of early work in this area, with a focus on randomized mechanisms and the above fairness notions, can be found in~\cite{AleksandrovW20}. Later on, \cite{benade2024fair} studied online allocation with a focus on minimizing envy over time while ensuring Pareto efficiency. They analyze trade-offs under varying levels of adversarial input—from worst-case to stochastic settings—and show that while no algorithm can simultaneously guarantee strong fairness and efficiency in the worst case, both objectives can be approximately achieved when item values follow known distributions.

A line of work starting by \cite{he2019achieving} (this particular work also allows limited re-allocations) relaxed the strong informational barrier on online models, by allowing algorithms to peak to the future. Peaking into the future, for full or partial information, has been used in further works~\cite{amanatidis2025online,neoh2025online} including the recent framework of \emph{temporal fair division}~\cite{cookson2025temporal,elkind2024temporal}. Under the last paradigm, the algorithm has full information of future arrivals, but the fairness guarantees need to hold up to each time step. Such requirements however are far reaching for EFX allocation, as it is already shown by~\cite{elkind2024temporal}.  

Other notable works on online fair division focus on different fairness guarantees, such as maximin share, EF1, or Nash welfare \cite{zhou2023multi,Song2025,wang2025online}. A distinct line of work considers the setting where agents, rather than goods, arrive online~\cite{kulkarni2025online,kash2014no}. Finally, we note that several works have considered online allocation of divisible resources~\cite{Gkatzelis_Psomas_Tan_2021,BanerjeeGHJM023,banerjee2022online}. 

A middle ground between full information for the future and no information at all is provided by the recently proposed framework of algorithmic design with predictive advice~\cite{LykourisV21,MitzenmacherV22}. Under this paradigm, the input of the algorithm comes also with some possibly unreliable prediction (e.g., about future inputs when we speak about online algorithms). An updating list of relevant papers is hosted in~\cite{lindermayr2024alps}. The idea has been used to augment fair division, though in markedly different settings. The work of \cite{banerjee2022online} uses prediction in an online fair division problem, where the goal is to allocate divisible resources in order to maximize the Nash welfare. \cite{BanerjeeGHJM023} use predictions on an online public goods setting, where goods appear online and a fraction of some budget is allocated to them. Finally, the work of~\cite{BSC23} explores the allocation of perishable goods with the help of predictions.

\section{Preliminaries}\label{sec:prelim}

For a positive integer $k$, we denote by $[k]$ the set $\{ 1, 2, \dots, k \}$. Given a non-empty bundle $A$ and a valuation function $f$, we denote $\xset{A}^f := A \setminus \{ g \}$, where $g \in \arg \max_{g' \in A} f(A \setminus \{g'\})$, while if $A = \emptyset$, then $\xset{A}^f := \emptyset$. Similarly, we denote $\oset{A}^f := A \setminus \{ g \}$, where $g \in \arg \min_{g' \in A} f(A \setminus \{g'\})$, and if $A = \emptyset$, then $\oset{A}^f := \emptyset$. For notation simplicity, instead of $f(\xset{A}^f)$ we write $f(\xset{A})$, and we omit the superscript when we have identical valuation functions, i.e., instead of $\xset{A}^f$ we write $\xset{A}$. The same notational simplification also applies to $\oset{A}^f$. We also simplify the notation when we refer to the valuation of a single good, and write $f(g)$ instead of $f(\{ g \})$. Function $f$ will be called \emph{$k$-value} if its codomain has cardinality at most $k$. In some of our results, we will use the golden ratio which we denote by $\varphi := \frac{1 + \sqrt{5}}{2} \approx 1.618$.

\paragraph{The model} 
We consider a set $N = [n]$ of $n \geq 2$ agents, and a set $M = \{ g_1, g_2, \dots, g_T \}$ of $T \geq 1$ goods. As it is common in fair division of indivisible goods, we will consider \emph{monotone} valuation functions. A valuation function $f : 2^M \to \preals$ is monotone if for any two sets $A, B \subseteq M$, we have $f(A \cup B) \geq f(A)$. We will especially focus on a natural subclass of monotone valuation functions, namely that of \emph{additive} valuations. A valuation function $f: M \to \preals$ is called additive if for any $A \subseteq M$, it holds that $f(A) = \sum_{g \in A} f(g)$. Importantly, we consider \emph{normalized} valuation functions, that is, $f(\emptyset) = 0$ and $f(M) = 1$. The normalization condition is crucial for our setting with predictions, a fact that will become apparent shortly. 

The setting is online, and involves discrete time-steps with a finite (true) horizon $T \geq 1$. Each agent $i$ first receives a \emph{prediction}, that is, a vector $p_i = (p_i(g_1), p_i(g_2), \dots, p_i(g_{T'}))$, where $T' \geq 1$ is the predicted horizon, and $p_i : M \to \preals$ is an additive, normalized valuation function, i.e., $p_{i}(g_t) \geq 0$ for all $t \in [T']$, and $\sum_{t \in [T']} p_{i}(g_{t}) = 1$. Then, at each time $t = 1, 2, \dots, T$, a single good $g_t$ arrives and has to be allocated \emph{irrevocably} to some agent. Each agent $i \in [n]$ at time $t$ evaluates good $g_t$ according to his additive, normalized \emph{true} valuation function $v_i : M \to \preals$, in other words, a vector $v_i = (v_{i}(g_1), v_{i}(g_2), \dots, v_{i}(g_T))$, where $v_{i}(g_t) \geq 0$ for all $t \in [T]$, and $\sum_{t \in [T]} v_{i}(g_{t}) = 1$. In the special case of \emph{identical valuations} every agent has the same prediction $p = (p(g_t))_{t \in [T']}$, and the same true valuation $v = (v(g_t))_{t \in [T]}$.

As a metric to quantify the distance between the predicted values and the true values of the goods, we use the \emph{Total Variation distance (TV distance)}.\footnote{Due to the normalization condition, for any $i \in [n]$, vectors $p_i$ and $v_i$ can be treated as probability mass functions, and as such, this distance metric can be applied on them, preserving all its properties.} In particular, this distance is described by the function 
\begin{align*}
    \TV{p_i}{v_i} = \left\| p_i -  v_i  \right\|_{\tvd} := \frac{1}{2} \sum_{t \in [\tmax]} \left| p_{i}(g_t) - v_{i}(g_t) \right| ,
\end{align*} 
where $\tmax := \max\{ T, T' \}$, and if $T' < T$, then $p_{i}(g_t) = 0$ for $t \in \{ T'+1, T'+2, \dots, T \}$, while if $T' > 0$, then $v_{i}(g_t) = 0$ for $t \in \{ T+1, T+2, \dots, T' \}$.\footnote{Notice that if $T' < T$, then the agent sees goods that were not predicted to exist. In that case, she might have extra \tvd distance from the last $T-T'$ time-steps, calculated by adding $T-T'$ ``dummy'' time-steps in the prediction and setting those prediction values to $0$. Similarly, if $T' > T$, then the agent expects goods that never appear. In that case, to calculate the \tvd distance, she adds $T'-T$ ``dummy'' time-steps and sets their true values to $0$.}
For ease of presentation, sometimes we refer to this distance as \emph{error}, and we denote it by $d_i := \TV{p_i}{v_i}$, and $D := \max_{i \in [n]} d_i$. Each agent $i$, has \emph{accuracy} $\acc_{i} := 1 - d_i$, which measures the fraction of the total value that is guaranteed to be predicted correctly by the agent. Therefore, if the accuracy is $100\%$, then the error is $0\%$, and the agent is capable of perfect predictions. In the other extreme where the accuracy is $0\%$, the error is $100\%$, and essentially this is equivalent to having no access to predictions.

\paragraph{Envy-freeness and some relaxed notions}
An \emph{allocation} $A = (A_1, A_2, \dots, A_n)$ of a set $M$ of goods to $n$ agents is a partition of $M$, where \emph{bundle} $A_i$ is allocated to agent $i \in [n]$. We say that agent $i$ \emph{envies} agent $j$ under her valuation $f_i$ if $f_{i}(A_i) < f_{i}(A_j)$, she \emph{envies up to any good (EFX-envies)} agent $j$ if $f_{i}(A_i) < f_{i}(\xset{A}_j)$, and she \emph{envies up to one good (EF1-envies)} agent $j$ if $f_{i}(A_i) < f_{i}(\oset{A}_j)$. An allocation is \emph{envy-free (EF)} if there is no envious agent, it is \emph{envy-free up to any good (EFX)} if there is no EFX-envious agent, and it is \emph{envy-free up to one good (EF1)} if there is no EF1-envious agent. The quantities $\max \{ f_{i}(A_j) - f_{i}(A_i) , 0 \}$, $\max \{ f_{i}(\xset{A}_j) - f_{i}(A_i) , 0 \}$, and $\max \{ f_{i}(\oset{A}_j) - f_{i}(A_i) , 0 \}$, are the \emph{envy}, \emph{EFX-envy}, and \emph{EF1-envy} of agent $i$ towards agent $j$, respectively. 

Apart from the exact versions of these fairness notions, we will study the approximate version of EFX and EF1. Under allocation $A$ and valuation $f_{i}$, agent $i$ \emph{$a$-envies up to any good (\aefx-envies)} agent $j$ if $f_{i}(A_i) < a \cdot f_{i}(\xset{A}_j)$ for some $a \in [0,1]$. Similarly, agent $i$ \emph{$a$-envies up to one good (\aefo-envies)} agent $j$ if $f_{i}(A_i) < a \cdot f_{i}(\oset{A}_j)$ for some $a \in [0,1]$. An allocation is \emph{$a$-approximately envy-free up to any good (\aefx)} or \emph{$a$-approximately envy-free up to one good (\aefo)} if there is no \aefx-envious agent or \aefo-envious agent, respectively. It is straightforward that a $1$-EFX or $1$-EF1 allocation is an (exact) EFX or EF1 allocation, respectively, and that any allocation is $0$-EFX and $0$-EF1. Throughout our results, it will be clear whether $f_{i}$ refers to the prediction $p_i$ or the true valuation $v_i$.

For some time-step $t$, we will also denote by $A^t := (A_1^t, A_2^t, \dots, A_n^t)$ the allocation (and the respective bundles to the agents) right after good $g_t$ gets allocated to some agent.

\section{Algorithms without Predictions}\label{sec:alg_without_predictions}

As a warm-up, we demonstrate what qualities of solutions are achievable without the use of predictions, and what qualities are impossible without predictions. 

We first consider the easier of the two relaxed envy-freeness notions, EF1, for general, identical valuations, without the use of predictions. It turns out that Theorem 3.7 of \cite{elkind2024temporal} intended for the temporal EF1 setting, can be applied in ours when we have $n \geq 2$ agents with monotone, identical valuations, without predictions, and even without normalization. This result shows that, for any true horizon $T$, the algorithm that allocates each arriving good to the agent with the lowest-valued bundle, provides at any given time $t \in [T]$ an exact EF1 allocation.

Recently, \cite{neoh2025online} provided an algorithm that computes an exact EF1 allocation for $2$ agents with normalized valuations that are not constrained to be identical, and they also showed that for $n \geq 3$ agents, no algorithm without predictions can achieve an \aefx for any $a \in (0, 1]$. In fact, concurrent with the drafting of our paper, the aforementioned manuscript appeared, and turned out to overlap with ours at the following results: \cref{thm:no_predictions_phi-EFX_inapprox_bound}, \cref{thm:no_predictions_phi-EFX_upper_bound}, \cref{thm:no_predictions_3_ag_0-EFX_inapprox}, and \cref{thm:no_predictions_non-id_2_ag_0-EFX_inapprox}. Since most of our proofs are different from theirs, we present them here for completeness.

Switching to the stronger relaxation of envy-freeness, namely EFX, even for two agents with additive, identical valuations satisfying the normalization condition, no algorithm can achieve an \aefx for $a \in (\varphi-1, 1]$ without access to predictions. 

\begin{theorem}\label{thm:no_predictions_phi-EFX_inapprox_bound}
    Suppose we have $2$ agents with additive, identical, normalized valuations, with no access to predictions. For any given $a \in (\varphi - 1, 1]$, there is no algorithm that guarantees an \aefx allocation, even if the time horizon is known.
\end{theorem}

\begin{proof}
    For the sake of contradiction, suppose there is an algorithm that guarantees an \aefx for some $a \in (\varphi - 1, 1]$. Consider a rational value $\lambda \in [0,  a - \varphi + 1)$. Let the adversary be giving to the algorithm goods of value $\varepsilon := \frac{\lambda}{4}$. Let the number of rounds be $T := \ceil{\frac{2\varphi - 3}{\varepsilon}} + 3$, known by the algorithm. W.l.o.g., the first good is allocated to agent 1. Now consider the time $t \geq 2$ when a good is allocated to agent 2. There are two cases:
    
    (i) If after that allocation, $v(A_1^t) =:x \leq 2 \varphi - 3$, then the adversary gives a good of value $1 - x - \varepsilon$ to the algorithm (and notice that $v(A_2^t) = \varepsilon$). If the latter good is allocated to agent 1, then this is an \praefx allocation with $a' \leq \frac{\varepsilon}{1 - 2 \varepsilon} = \frac{\lambda}{4 - 2 \lambda} < a$, by definition of $\lambda$. This is a contradiction. If the latter good is allocated to agent 2, then this is an \praefx allocation with $a' \leq \frac{x}{1 - x - \varepsilon} \leq \frac{2 \varphi - 3}{4 - 2 \varphi - \varepsilon} < \frac{2 \varphi - 3}{4 - 2 \varphi - (2 - \varphi)} = \frac{2 \varphi - 3}{2 - \varphi} = \varphi - 1 < a$, where the third inequality comes from the fact that $\varepsilon < \lambda \leq 2 - \varphi$. So in this case the algorithm cannot provide an \aefx allocation for $a > \varphi$.

    (ii) Otherwise, there exists a time $t'$ where $v(A_1^{t'}) =: x \in \left( 2 \varphi - 3, 2 \varphi - 3 + \varepsilon \right]$ and $v(A_2^{t'}) = 0$. Then, the adversary gives consecutively two goods, each of value $\frac{1-x}{2}$. If both are allocated to agent 1, then this is obviously a $0$-EFX, a contradiction. If both are allocated to agent 2, then agent 2 does not envy agent 1, but agent 1 has EFX-envy towards agent 2, and this \praefx allocation has $a' \leq \frac{x}{(1-x)/2} \leq \frac{2(2 \varphi - 3 + \varepsilon)}{4 - 2 \varphi - \varepsilon} \leq \frac{2(2 \varphi - 3) + \lambda/2}{4 - 2 \varphi - \lambda/4} < \varphi - 1 + \lambda < a$, where the second to last inequality comes from the fact that $\lambda < a - \varphi + 1 \leq 2 - \varphi$. So this is not an \aefx allocation. Finally, if one good is allocated to agent 1 and the other to agent 2, then agent 1 does not envy agent 2, but agent 2 has EFX-envy towards agent 1, and this \praefx allocation has $a' \leq \frac{(1-x)/2}{x + (1-x)/2 - \varepsilon} = \frac{1 - x}{1 + x - 2 \varepsilon} \leq \frac{4 - 2 \varphi}{2 \varphi - 2 - 2 \varepsilon} = \frac{4 - 2 \varphi}{2 \varphi - 2 - \lambda/2} < \varphi - 1 + \lambda < a$, where again the second to last inequality comes from the fact that $\lambda < 2 - \varphi$. So the algorithm in this case cannot produce an \aefx allocation. Therefore, for any fixed $a \in (\varphi - 1, 1]$, an \aefx allocation is impossible to guarantee.

    From the case analysis above, we see that the maximum number of rounds required is for case (ii), and it is $\ceil{\frac{2\varphi - 3 + \varepsilon}{\varepsilon}} + 2$. For all the cases that the algorithm might fall into, the adversary provides as many extra goods as needed to complete the known $T = \ceil{\frac{2\varphi - 3 + \varepsilon}{\varepsilon}} + 2$ rounds, where all these goods' values are set to $0$ and do not affect the bounds on the approximation ratios. 
\end{proof}

Moreover, there is a simple greedy algorithm that matches this bound on $a$ without even the need to access predictions.

\begin{theorem}\label{thm:no_predictions_phi-EFX_upper_bound}
    Suppose we have $2$ agents with additive, identical, normalized valuations, and without predictions. Consider the following algorithm: allocate each arriving good to agent 1, as long as the good will not make her exceed value $\varphi - 1$, otherwise allocate it to agent 2. This algorithm guarantees an allocation which is $(\varphi - 1)$-EFX.
\end{theorem}

\begin{proof}
    When the algorithm terminates, in the resulting allocation $A$, agent 1 will have value $x := v(A_1) \leq \varphi - 1$, by the algorithm's definition. Agent 2 will have value $y := v(A_2) = 1-x \geq 2 - \varphi$. If $x \geq y$, then this \aefx allocation has $a \geq \frac{y}{x} \geq \frac{2-\varphi}{\varphi - 1} = \varphi - 1$. If $x < y$, then all goods of agent 2 have value strictly greater than $\varphi -1 - x$, otherwise, at least one of them would have been given to agent 1 (since together with the new good agent 1 would have value at most $x + (\varphi - 1 - x) = \varphi - 1$). If agent 2 has a single good, then agent 1 does not EFX-envy agent 2, and agent 2 does not envy agent 1, since $y > x$, so this is an exact EFX. If agent 2 has at least two goods, then each has value greater than $\varphi - 1 - x$ as argued earlier, so $1-x = y > 2(\varphi - 1 - x)$ which implies that $x > 2 \varphi - 3$. Then, in this \aefx allocation, $y > x$ so agent 2 does not envy agent 1, and also $v(\xset{A}_2) < y - (\varphi - 1 - x) = 1 - x - (\varphi - 1 - x)  = 2 - \varphi$. So, we have $a \geq \frac{2 \varphi - 3}{2 - \varphi} = \varphi - 1$. 
\end{proof}

By increasing the number of agents to three or more, the approximation of EFX becomes impossible without the use of predictions, even when their valuations are identical.

\begin{theorem}\label{thm:no_predictions_3_ag_0-EFX_inapprox}
    Suppose we have $n \geq 3$ agents with additive, identical, normalized valuations, and without predictions. For any given $a \in (0, 1]$, there is no algorithm that guarantees an \aefx allocation, even if the time horizon is known.
\end{theorem}

\begin{proof}
    Consider an algorithm that knows the true horizon $T = n+1$ and does not use predictions. Let the first two goods that the adversary gives to the algorithm have value $\varepsilon := \frac{a}{3 (n-1)}$. There are two cases: 
    
    (i) If the algorithm allocates the two goods to the same agent, say $i$, then the adversary provides a third good of value $1 - 2 \varepsilon$, and another $n-2$ goods of value 0.\footnote{The latter $n-2$ goods of value $0$ are given so that the adversary fulfils the promise that $T=n+1$, and do not affect the bounds on the EFX approximation ratio.} Regardless of where the latter $n-1$ goods are allocated, there is an agent with value $0$ who \aefx-envies agent $i$ for any $a > 0$, so in this case, the algorithm does not guarantee an \aefx allocation. 
    
    (ii) If the algorithm allocates the first two goods to two agents, say $i, j$, then the adversary provides $n-1$ goods of value $\frac{1-2\varepsilon}{n-1}$. Then, by the pigeonhole principle, at least one of the agents will have at least two goods (one of which has value $\frac{1-2\varepsilon}{n-1}$), and at least one of them will have value at most $\varepsilon$. Therefore, one of the latter agents is \praefx free for $a' \leq \frac{\varepsilon}{(1-2\varepsilon)/(n-1)} < \frac{\varepsilon}{1/(2 (n-1))} = 2 \varepsilon (n-1) < a$, where the second inequality comes from the fact that $1 - 2 \varepsilon = 1 - \frac{2a}{3 (n-1)} > 1/2$ for any $a \leq 1$ and $n \geq 3$. Therefore, in this case too, the algorithm cannot give an \aefx allocation.
\end{proof}

When we alleviate the ``identical valuations'' restriction, the approximation of EFX without predictions becomes impossible, even for two agents.

\begin{theorem}\label{thm:no_predictions_non-id_2_ag_0-EFX_inapprox}
    Suppose we have 2 agents with additive, normalized valuations, without predictions. For any given $a \in (0, 1]$, there is no algorithm that guarantees an \aefx allocation, even if the time horizon is known.
\end{theorem}

\begin{proof}
    For the sake of contradiction, suppose there is an algorithm that can provide an \aefx allocation for some $a > 0$. The algorithm also knows that $T = \floor{\frac{4}{a}}+2$. We will denote the value of a good $g$ as $(\ell, r)$, where $\ell := v_1 ( g )$ and $r := v_2 ( g )$. We fix an $\varepsilon := \frac{a
    }{4}$. Let the first good that the adversary gives to the algorithm be of value $(\varepsilon, 0)$. There are two cases:

    (i) If the algorithm allocates the first good to agent 2, then the adversary keeps giving to the algorithm goods of value $(\varepsilon, 0)$ until the algorithm allocates a good to agent 1 (which might never happen, which means that after $t=T$, agent 1 has no goods, and this is a $0$-EFX allocation). Let time $t \geq 2$ be the round in which either a good is allocated to agent 1 for the first time or when $t \varepsilon \in  (1 - \varepsilon, 1]$. Observe that if the latter condition is the case, then no more goods of value $(\varepsilon, 0)$ can be produced by the adversary, otherwise, according to agent 1, the total value for all the goods produced would exceed the value of $1$, thus defy the normalization. We know that $v_1 (A_1^t) \leq \varepsilon$ and $v_1 (A_2^t) \geq (t-1) \varepsilon$. Then, the adversary produces a good of value $(1 - t \varepsilon, 1)$. If the algorithm allocates the latter good to agent 2, then agent 2 does not EFX-envy agent 1, but agent 1 has EFX-envy towards agent 2, and the \praefx allocation has $a' \leq \frac{v_1 (A_1^{t+1})}{v_1 (A_2^{t+1}) - \varepsilon} \leq \frac{\varepsilon}{(1 - t \varepsilon) + (t-1) \varepsilon - \varepsilon} = \frac{\varepsilon}{1 - 2 \varepsilon} \leq \frac{a}{2} < a$, where the penultimate inequality comes from the fact that $1 - 2 \varepsilon = 1 - \frac{a}{2} \geq \frac{1}{2}$ for any $a \leq 1$. If the algorithm allocates the good to agent 1, then agent 2 has EFX-envy towards agent 1, and this can only be a $0$-EFX allocation since $v_2 (A_2^{t+1}) = 0$ and $v_{2}(\xset{A}_1^{t+1}) = 1$. So in this case, the algorithm cannot provide an \aefx allocation.

    (ii) If the algorithm allocates the first good to agent 1, then the adversary keeps giving to the algorithm goods of value $(0, \varepsilon)$ until the algorithm allocates a good to agent 2 (which might never happen). Similarly to case (i), let time $t \geq 2$ be the round in which either a good is allocated to agent 2 for the first time or when $(t-1) \varepsilon \in  (1 - \varepsilon, 1]$. Observe again that in the latter condition, for a reason symmetric to that of case (i), no more goods of value $(0, \varepsilon)$ can be produced by the adversary. We know that $v_2 (A_2^t) \leq \varepsilon$ and $v_2 (A_1^t) \geq (t-2) \varepsilon$. Then, the adversary produces a good of value $(1 - \varepsilon, 1 - (t-1) \varepsilon)$. If the algorithm allocates the latter good to agent 1, then agent 1 does not EFX-envy agent 2, but agent 2 has EFX-envy towards agent 1, and the \praefx allocation has $a' \leq \frac{v_2 (A_2^{t+1})}{v_2 (A_1^{t+1}) - 0} \leq \frac{\varepsilon}{(1 - (t-1) \varepsilon) + (t-2) \varepsilon} = \frac{\varepsilon}{1 - \varepsilon} \leq \frac{a}{3} < a$, where the penultimate inequality comes from the fact that $1 - \varepsilon = 1 - \frac{a}{4} \geq \frac{3}{4}$ for any $a \leq 1$. If the algorithm allocates the good to agent 2, then agent 1 has EFX-envy towards agent 2, and the \praefx allocation has $a' \leq \frac{\varepsilon}{1 - \varepsilon} < a$ as argued earlier. So in this case too, the algorithm cannot provide an \aefx allocation.

    From the case analysis above, we see that the maximum number of goods produced is $1 + t + 1 = \floor{\frac{1}{\varepsilon}}+2 = \floor{\frac{4}{a}}+2$, which is needed for case (ii). If the algorithm falls in any of the two cases and the aforementioned number of rounds has not been reached, the adversary provides the algorithm with an extra amount of goods until that number of rounds is achieved (in order to fulfill its promise that $T = \floor{\frac{4}{a}}+2$). All the extra goods have value $(0,0)$ and do not affect the bounds on the approximation ratios.
\end{proof}

\section{Algorithms with Predictions}\label{sec:alg_with_predictions}

The prediction model we consider can capture various real-life applications where goods whose value is not fully predictable arrive over time and have to be allocated to some agent once and for all. In particular, the error as we measure it can reflect the following scenarios, among others: 
\begin{enumerate}
    \item[(i)] The agent knows which good will arrive in each round, but incorrectly predicts its relative value, with the error being solely in the value prediction.
    \item[(ii)] The agent knows which good will arrive in each round, but upon arrival, may misjudge its quality.
    \item[(iii)] A combination of the above, where the agent not only mispredicts the value but also misjudges the good's quality upon arrival.
    \item[(iv)] The agent knows the set of goods to arrive, but incorrectly predicts their order of arrival.
    \item[(v)] The agent knows the set of goods to arrive, but makes a combination of errors: she might mispredict their arrival order, mispredict their relative value, misjudge their quality upon arrival, or all three.
    \item[(vi)] The agent is completely unaware of which goods will arrive, and additionally, mispredicts or misjudges according to some of the above scenarios.
\end{enumerate}
For example, an instance of scenario (i) with value-prediction error $30 \%$, and an instance of scenario (iii) with value-prediction error $10 \%$ and judgement error $20 \%$, can both be captured by our model with prediction error $30 \%$. This means that, if an algorithm guarantees an \aefx (respectively, \aefo) allocation for some $a \in (0, 1]$ when the error is $0.3$, then it can give a solution to both instances, while if it is impossible for such an algorithm to exist for one of these instances, then the same holds for the other instance too.

As we have seen so far, there are sharp dichotomies on the approximation factor of EF1 and EFX allocations when no predictions are involved. In particular, for identical valuations, $1$-EF1 exists for $n \geq 2$ agents; $(\varphi - 1)$-EFX exists for $n=2$ and no \aefx exists for $a \in (\varphi - 1, 1]$, while for $n \geq 3$ no \aefx exists for $a \in (0,1]$. When the valuation functions are not restricted to be identical, the impossibility results become significantly stronger: $1$-EF1 exists for $n=2$, but when $n \geq 3$, no \aefo exists for any $a \in (0, 1]$; for $n \geq 2$, and no \aefx exists for any $a \in (0, 1]$. This means that to improve the approximation quality of EF1 and EFX allocations, predictions have to be involved. It is important to note here that when predictions are involved, each agent's accuracy is known to the algorithm (can be part of the input).

\subsection{Full reliance on predictions}\label{sec:alg_with_predictions_entirely}

We start by exploring the capabilities of algorithms that rely entirely on predictions.

\begin{theorem}\label{thm:with_predictions_n_ag_a_i-EFX_positive}
    Suppose we have $n \geq 2$ agents with additive, normalized valuations. Each agent $i \in [n]$ has accuracy $\eta_i \geq 1 - \frac{\tila_i-a_i}{(2n-2+\tila_i)(1+a_i)}$ for some given $a_i, \tila_i \in \left[ 0, 1 \right]$ with $a_i \leq \tila_i$, that is, for agent $i$ the error between the provided prediction $p_i$ and the true valuation $v_i$ is $d_i = 1 - \eta_i \leq \frac{\tila_i-a_i}{(2n-2+\tila_i)(1+a_i)}$. Given an allocation $A^o$ in which every agent $i$ has no $\tila_i$-EFX-envy towards any other agent according to $(p_{i}(g_t))_{t \in [T']}$, we can compute in polynomial time an allocation $B$ in which every agent $i$ has no $a_i$-EFX-envy towards any other agent according to her true valuation $(v_{i}(g_t))_{t \in [T]}$.
\end{theorem}

\begin{proof}
    Let the allocation $A^o = (A_1^o, A_2^o \dots, A_n^o)$ be the output of the statement's algorithm when its input is the prediction $(p_{i,t})_{i \in [n], t \in [T']}$, therefore in $A^o$ each agent $i \in [n]$ is not $\tila_i$-EFX-envious towards any other agent according to the prediction values. We will first turn $A^o$ into an allocation $A$ in which there is at least one agent $k \in [n]$ that is not envied by any other agent (if this is not already the case). We define the ``envy-graph'' of $A^o$; this is a directed graph whose nodes are the agents, and an edge $(i,j)$ exists if $i$ envies $j$. If the graph has a source, then this source is our required $k$ and we are done. If the graph has no source, it contains a cycle $i_1, i_2, \dots, i_r, i_1$ of size $r$ for some $r \in \{2, 3 \dots, n\}$. We can eliminate this cycle by rotating the bundles, i.e., $A_{i_s}^o \gets A_{i_{s+1}}^o$ for all $s \in [r]$, where $i_{r+1} := i_{1}$. Notice that by rotating the bundles, we have eliminated $r \geq 2$ edges of the graph and did not create any edges, since the bundles of $A^o$ remained the same even though they changed owners. By eliminating cycles repeatedly until there are no more left, we derive an allocation $A$ whose envy-graph is acyclic and therefore has a source, that is, the required agent $k$ who is not envied by any other agent. Finally, since the edges of the initial envy-graph were at most $n(n-1)$ and in each cycle elimination we strictly reduced the number of edges by 2, allocation $A$ will be found after at most $n(n-1)/2$ eliminations. Creating the envy-graph of $A^o$ and each of its updates takes time polynomial in $n$ and $T$, so computing $A$ can be done in polynomial time.
    
    Now let us focus on an arbitrary agent $i$ under allocation $A$. We denote $x_{i*} := p_{i}(A_i)$, and $x_{i,j} := p_{i}(\xset{A}_j)$ for $j \neq i$. Since agent $i$ is not $\tila_i$-EFX-envious towards any other agent, we have 
    \begin{align}\label{eq:a_i-EFX_positive}
        x_{i*} \geq \tila_i \cdot x_{ij}, \quad \text{for all $j \neq i$}. 
    \end{align}
    We also know that for any $i \neq j$, $x_{ij} \geq \frac{p_{i}(A_j)}{2}$: If $A_j = \emptyset$ then this is obviously true, and if $A_j \neq \emptyset$ then, by setting $g \in \arg \min_{g \in A_j} p_{i}(A_j)$, we have $p_{i}(g) \leq \min_{g' \in A_j \setminus \{g\}} \{ p_{i}(g') \} \leq x_{ij}$, so $p_{i}(A_j) = p_{i}(g) + x_{ij} \leq 2 \cdot x_{ij}$ for any $j \neq i$. Finally, by the normalization assumption, we have $1 = \sum_{j \in [n]} p_{i}(A_j) \leq x_{i*} + (n-1) \cdot \frac{2}{\tila_i} x_{i*} $, or equivalently, $x_{i*} \geq \frac{\tila_i}{2n-2+\tila_i}$.

    Now we modify $A$ to get allocation $B$ in the following way, depending on how the true time horizon $T$ compares to the predicted horizon $T'$.
     \begin{align*}
         B_i = \begin{cases}
			A_i \setminus \{ g_{T+1}, \dots, g_{T'} \}, & \text{if $T \leq T'$, for all $i \in [n]$}\\
            A_i, & \text{if $T > T'$, for all $i \neq k$}\\
            A_i \cup \{ g_{T'+1}, \dots, g_{T} \}, & \text{if $T > T'$, for $i = k$}.
		 \end{cases}
     \end{align*}
     In other words, if the number of goods that arrive is at most that of the prediction, then the goods that arrive are allocated according to the prediction. If more goods arrive than those predicted, then they are placed in the bundle of agent $k$ who was not envied by anyone in allocation $A$. 
    
    Consider the true values' distribution $(v_{i,t})_{i \in [n], t \in [T]}$, and similarly to the previous notation, let us denote $y_{i*} := v_{i}(B_i)$, and  $y_{i,j} := v_{i}(\xset{B}_j)$ for $j \neq i$.    
    For some fixed $a_i \in [0, 1]$ with $a_i \leq \tila_i$, recall that $(v_{i,t})_{i \in [n], t \in [T]}$ has \tvd distance $d_i \leq \frac{\tila_i-a_i}{(2n-2+\tila_i)(1+a_i)}$ from $(p_{i,t})_{i \in [n], t \in [T']}$. Notice that $v_{i}(B_i) \geq x_{i*} - d_i$, so $y_{i*} \geq x_{i*} - d_i$. Also, for every $j \in [n] \setminus \{ k \}$ we have $ y_{ij} \leq  x_{ij} + d_i \leq \frac{x_{i*}}{\tila_i} + d_i $, where the last inequality comes from \cref{eq:a_i-EFX_positive}. Finally, since agent $k$ is not envied by any other agent in allocation $A$, we have $x_{i*} \geq p_{i}(A_k)$. We then have
    $y_{ik} \leq p_{i}(A_k) + d_i \leq x_{i*} + d_i \leq \frac{x_{i*}}{\tila_i} + d_i$, where the last inequality is due to the fact that $\tila_i \leq 1$.

    Then in $B$, agent $i$ is not $a'_i$-EFX-envious towards any other agent for 
    \begin{align*}
        a'_i = \frac{y_{i*}}{\max_{i \in [n]} \{ y_{ij} \} } \geq \frac{x_{i*} - d_i}{x_{i*}/\tila_i + d_i} \geq \frac{\tila_i - (2n - 2 + \tila_i) d_i}{1 + (2n - 2 + \tila_i) d_i} \geq \frac{\tila_i - (\tila_i-a_i)/(1+a_i)}{1 + (\tila_i-a_i)/(1+a_i)} = a_i.
    \end{align*}
\end{proof}

\begin{corollary}\label{cor:with_predictions_n_ag_a-EFX_positive}
    Suppose we have $n \geq 2$ agents with additive, normalized valuations. The agents have accuracy $\eta \geq 1 - \frac{\tila - a}{(2n-2+\tila)(1+a)}$ for some given $a, \tila \in \left[ 0, 1 \right]$ with $a \leq \tila$, that is, for any agent $i \in [n]$, the error between the provided prediction $p_i$ and the true valuation $v_i$ is at most $D = 1 - \eta \leq \frac{\tila - a}{(2n-2+\tila)(1+a)}$. Given an $\tila$-EFX allocation $A^o$ according to $(p_{i}(g_t))_{i \in [n], t \in [T']}$, we can compute in polynomial time an allocation $B$ which is \aefx according to the true valuations $(v_{i}(g_t))_{i \in [n], t \in [T]}$.
\end{corollary}

The proof of \cref{thm:with_predictions_n_ag_a_i-EFX_positive} reveals that the algorithm used to produce an $a$-EFX allocation (according to the true valuations) is an algorithm oblivious to the true valuations themselves. In particular, first it takes as input an $\tila$-EFX allocation based on the predicted values, then shifts bundles around such that there is an agent who is not envied by any other, and then takes that prescribed allocation and follows it blindly on the true values (were if unpredicted goods appear, they get allocated to the unenvied agent). Notice that, for several classes of valuations we have algorithms that guarantee an exact EFX allocation, that is $\tila = 1$. For example, for identically-ordered additive valuations \cite{plaut2020almost}, when there are at most two monotone (not necessarily additive) valuations the agents can choose from~\cite{Mahara24}, for $n=3$ with additive valuations~\cite{CGM24}, for monotone valuations obeying a graph structure~\cite{Christodoulou2023}, and many more (see \cite{amanatidis2023fair}). Outside of such cases, $\tila = \varphi - 1$ can always be achieved for additive valuations~\cite{amanatidis2020multiple}, and even $\tila = 2/3$ for $n \leq 7$~\cite{AFS24} or when agents agree on their $n$ favourite goods~\cite{markakis23improved}. 

We can actually show that when it is possible to compute an exact EFX allocation with respect to the predicted values, the accuracy bound of \cref{cor:with_predictions_n_ag_a-EFX_positive} is tight among algorithms that rely entirely on predictions. In other words, if it is possible to compute an $1$-EFX allocation for the predictions, for the algorithms that follow any $1$-EFX derived by the predictions and ignore the true values, the level of accuracy stated in the aforementioned result for $\tila=1$, i.e. $1 - \frac{1-a}{(2n-1)(1+a)}$, is necessary to achieve an $a$-EFX allocation. We show that this is true even for identical valuations.

\begin{proposition}\label{prop:no_true_vals_inapprox}
    Suppose we have $n \geq 2$ agents with additive, identical, normalized valuations. The agents have accuracy $\eta < 1 - \frac{1 - a}{(2n-1)(1+a)}$ for some given $a \in \left[ 0, 1 \right]$, that is, the error between the provided prediction $(p(g_t))_{t \in [T']}$ and the true values $(v(g_t))_{t \in [T]}$ is $D = 1 - \eta > \frac{1 - a}{(2n-1)(1+a)}$. Then, any algorithm that computes an exact EFX allocation on the predictions and follows it by ignoring the true values, is not able to compute an \aefx allocation according to the true valuation, even when $T' = T$.
\end{proposition}

\begin{proof}
    Consider the following instance with identical valuations, for which we can compute an exact EFX allocation on the predictions, by using LPT (\cref{alg:LPT}). Let $T' = 2n-1$, and where $p(g_t) = \frac{1}{2n-1}$ for all $t \in [2n-1]$. Then, the only $1$-EFX allocation is one where all agents receive 2 goods, except for one, w.l.o.g. agent 1, who receives 1 good. Consider an adversary producing $T = 2n-1$ true values $v(g_1) = \frac{1}{2n-1} - D < \frac{2a}{(2n-1)(1+a)}$, $v(g_2) = \frac{1}{2n-1} + D > \frac{2}{(2n-1)(1+a)}$, $v(g_t) = \frac{1}{2n-1}$ for $t \in \{ 3, 4, \dots, 2n-1 \}$. The algorithm ignores the true values and assigns the goods according to the aforementioned $1$-EFX allocation based only on the predictions. Now $v(A_1) = v(g_1) < \frac{1}{2n-1} - \frac{1-a}{(2n-1)(1+a)} = \frac{2a}{(2n-1)(1+a)}$, and $v(\xset{A}_2) = v(g_2) > \frac{1}{2n-1} + \frac{1-a}{(2n-1)(1+a)} = \frac{2}{(2n-1)(1+a)}$, therefore this is an \praefx allocation for $a' \leq \frac{v(A_1)}{v(\xset{A}_2)} < a$. So the algorithm fails to produce an \aefx allocation.
%
\end{proof}

Before proceeding to show an example of \cref{cor:with_predictions_n_ag_a-EFX_positive}, we present a result -- considered to be folklore by now -- for efficiently computing exact EFX allocations for $n \geq 2$ agents with additive, identical valuations. \cref{alg:LPT} is a well-known algorithm, called \emph{longest-processing-time-first (LPT)}, that comes from the literature of job-scheduling over identical machines. There, the problem is to schedule jobs (here a job is represented by a good $g_t, t \in [T]$), each coming with a known processing time (here, $f(g_t)$), into machines (here, bundles $A_i$ of agents $i \in [n]$) and the goal is to minimize the maximum processing time of a machine (here, $f(A_i)$) over all machines. Even though it has a good approximation guarantee, the LPT algorithm does not find the optimum for the job-scheduling problem. However, it is easy to show that in the context of fair division of indivisible goods under an additive valuation, it provides an exact EFX allocation.\footnote{We note that a generalized version of this algorithm has been used in \cite{GoldbergHH23} to find EFX allocations for a superclass of additive valuations. For the sake of completeness, we present here a self-sustained proof for additive valuations.} The LPT algorithm is a central component of \cref{alg:id_2_ag}.

\begin{algorithm}[tbh]
\caption{(LPT) Compute an (offline) exact EFX allocation for $n \geq 2$ agents with additive, identical valuations}\label{alg:LPT}  
	\begin{algorithmic}[1]
		\REQUIRE{A valuation function $f$ over set $M = \{g_1, g_2, \dots, g_T\}$ for $n \geq 2$ agents.}
		\ENSURE{An EFX allocation $A$.}
		
		\medskip
		
		\STATE{$(A_1, A_2, \dots, A_n) \gets (\emptyset, \emptyset, \dots, \emptyset)$}
		
        \medskip
		
		\STATE{Rename goods of $M$ according to non-ascending order of value $f$, creating $M' := \{ g'_1, g'_2, \dots, g'_T \}$}\label{alg:LPT-sort-line}
        \FOR{$t \in [T]$}
            \STATE{Find an agent $i^*$ with $f(A_{i^*}) = \min_{i \in [n]} \{ f(A_i) \}$}\label{alg:LPT-min_line}
            \STATE{$A_{i^*} \gets A_{i^*} \cup \{ g'_t \}$}\label{alg:LPT-alloc_line}
        \ENDFOR
        \RETURN $A$
    \end{algorithmic}
\end{algorithm}

\begin{proposition}\label{prop:LPT}
    For $n \geq 2$ agents with the same additive valuation function $f : [T] \to \preals$ given explicitly in the input, the LPT algorithm (\cref{alg:LPT}) outputs an exact EFX allocation in time $O(T \log (nT))$. 
\end{proposition}

\begin{proof}
    Notice that the algorithm works in the offline setting, since it needs to have as input the values $f(g_1), f(g_2), \dots, f(g_T)$. By induction on $t$, we will prove that after creating the set $M'$ of the goods according to a non-ascending order of value, after each round $t \in [T]$, the partial allocation is an exact EFX allocation. In round $t=1$, an arbitrary agent has been allocated the most valuable good, and this is obviously an exact EFX allocation. Now suppose that we have an exact EFX allocation at time $t^* - 1$ for some $t^* \in \{2, 3, \dots, T-1 \}$; we will show that the algorithm gives an exact EFX allocation at time $t^*$. Let the for-loop of the algorithm be at round $t^*$. Note that agent $i^*$ (who receives a good in this round) cannot increase her EFX-envy towards other agents. So it suffices to show that the rest of the agents do not EFX-envy $i^*$. At line \ref{alg:LPT-min_line}, no agent envied $i^*$. Also, by definition of $M'$, we have $f(g'_{t^*}) \leq f(g'_{t})$ for all $t \in [t^*]$, so at line \ref{alg:LPT-alloc_line} agent $i^*$ gets allocated the good of minimum value among the ones already allocated to the agents. Therefore, for every agent $i \in [n]$ we have $f(A_i) \geq f(\xset{A}_{i^*})$, so no one EFX-envies $i^*$. Notice that \ref{alg:LPT-sort-line} takes $O(T \log T)$ time to sort $T$ values, and line \ref{alg:LPT-min_line} takes $O(\log n)$ time (e.g., by using a min-heap), resulting in the for-loop having time $O(T \log n)$.
\end{proof}

To showcase the contribution of predictions to the quality of EFX allocations, let us see the following instantiation of \cref{cor:with_predictions_n_ag_a-EFX_positive}. We have the simple case of two agents with additive, identical valuations. As \cref{thm:no_predictions_phi-EFX_inapprox_bound} shows, an algorithm without access to predictions cannot guarantee an $a$-EFX for any $a \geq 0.619$. In contrast, when we are allowed prediction of accuracy $94.5 \%$, \cref{alg:LPT} can compute a $0.718$-EFX allocation in polynomial time.

\begin{example}\label{ex:LPT}
    For two agents with additive, identical, normalized valuations, LPT (\cref{alg:LPT}) guarantees an exact EFX (that is, an $1$-EFX) on the predicted values in polynomial time. If, additionally, the agents have prediction accuracy $\eta \geq 1 - \frac{1-(\varphi - 0.9)}{(2 \cdot 2 - 2 + 1)(1 + \varphi - 0.9)} \simeq 0.945$, then LPT's allocation is $(\varphi - 0.9) \simeq 0.718$-EFX on the true values.
\end{example}

\subsection{Use of predictions and true values}\label{sec:alg_with_predictions_and_true}

The limitations derived so far, show that, unless both predictions and the observed (true) values are taken into consideration by an algorithm, the accuracy required to output an \aefx allocation for a given $a \in [0,1]$, is relatively high, even for identical valuations (see \cref{prop:no_true_vals_inapprox}). To bridge the gap between the lower and the upper bound on the accuracy, we turn our focus on algorithms that use the information of the predictions to decide where to allocate each arriving good.

\subsubsection{Impossibility results}

We first show lower bounds on the necessary level of accuracy, which essentially tell us that no algorithms to compute \aefx exist when their accuracy is below a threshold. Recall from \cref{thm:no_predictions_non-id_2_ag_0-EFX_inapprox} that, without predictions, no algorithm can guarantee an \aefx allocation for any $a \in (0,1]$, even for two agents. The following lower bound shows that it remains impossible to compute an \aefx for any $a \in \left( \frac{1}{2}, 1 \right]$ even for algorithms enhanced with predictions of accuracy less than a particular threshold. For the case of $a \in \left(0, \frac{1}{2} \right)$ we know from \cref{thm:no_predictions_non-id_2_ag_0-EFX_inapprox} that predictions (i.e., positive accuracies) are necessary to get an allocation better than $0$-EFX, but it is an open question whether there is a positive lower bound on the accuracy.

\begin{theorem}\label{thm:with_predictions_non-id_2_ag_a-EFX_inapprox}
    Suppose we have $2$ agents with additive, normalized valuations, with a provided prediction of accuracy $\eta < 1 - \frac{1-a}{\min\{ 6a, 4 \}}$ for some given $a \in \left( \frac{1}{2}, 1 \right]$, that is, the error between the predictions and the true valuations is $1 - \eta > \frac{1-a}{\min\{ 6a, 4 \}}$. Then, there is no algorithm that guarantees an \aefx allocation, even when $T'=T=4$, and the predictions and the true valuations are $4$-value functions.
\end{theorem}

\begin{proof}
    Consider an arbitrary $a \in \left( \frac{1}{2}, 1 \right]$ and an \aefx algorithm with prediction accuracy $\eta < 1 - \frac{1-a}{\min\{ 6a, 4 \}}$. The adversary fixes a rational number
    $\lambda \in \begin{cases}
			\left[0, \frac{2a-1}{6a} \right), & \text{if $a \in \left(\frac{1}{2}, \frac{2}{3} \right]$}\\
            \left(\frac{1-a}{4}, \frac{a}{8} \right), & \text{if $a \in \left(\frac{2}{3}, 1 \right]$}
		 \end{cases}$,
    and sets $\varepsilon = \begin{cases}
			\frac{1}{6} - \lambda, & \text{if $a \in \left(\frac{1}{2}, \frac{2}{3} \right]$}\\
            \lambda, & \text{if $a \in \left(\frac{2}{3}, 1 \right]$}
		 \end{cases}$.
    Notice that for both $a \in \left(\frac{1}{2}, \frac{2}{3} \right]$ and $a \in \left(\frac{2}{3}, 1 \right]$, we have
    \begin{align}\label{eq:bounds_eps}
        \varepsilon \in \left( \left( \frac{1}{2} - \lambda \right) \cdot \frac{1-a}{1+a}, \quad \frac{a}{2(2-a)} - \lambda \cdot \frac{2+a}{2-a} \right) 
    \end{align}
    due to the domain of $\lambda$.
    The adversary also provides the following prediction to the algorithm for the two agents: $p_1(g_1) = 2 \varepsilon$, $p_1(g_2) = 2 \lambda$, $p_1(g_3) = \frac{1}{2} - 2 \varepsilon - \lambda$, $p_1(g_4) = \frac{1}{2} - \lambda$, and $p_2(g_1) = 2 \lambda$, $p_2(g_2) = 2 \varepsilon$, $p_2(g_3) = \frac{1}{2} - 2 \varepsilon - \lambda$, $p_2(g_4) = \frac{1}{2} - \lambda$. The adversary then reveals the true values $v_1(g_1) = 2 \varepsilon$, $v_1(g_2) = 2 \lambda$, and $v_2(g_1) = 2 \lambda$, $v_2(g_2) = 2 \varepsilon$. There are four cases:

    (i) If $g_1$ and $g_2$ get allocated to agent 1, then the adversary reveals true values $v_1(g_3) = \frac{1}{2} - 3 \varepsilon - \lambda$, $v_1(g_4) = \frac{1}{2} + \varepsilon - \lambda$, and $v_2(g_3) = \frac{1}{2} - \varepsilon - \lambda$, $v_2(g_4) = \frac{1}{2} - \varepsilon - \lambda$. Then, if both $g_3, g_4$ get allocated to agent 1, then agent 2 is EFX-envious towards agent 1, and this is a $0$-EFX allocation. If both goods get allocated to agent 2, then agent 1 is EFX-envious to agent 2 and this is an \praefx allocation with $a' = \frac{2 \varepsilon + 2 \lambda}{1/2 + \varepsilon - \lambda} < a$, where the inequality comes from the upper bound of $\varepsilon$ in \cref{eq:bounds_eps}. If $g_3$ gets allocated to agent 1 and $g_4$ to agent 2, then this is an \praefx allocation with $a' \leq \frac{1/2 - \varepsilon - \lambda}{1/2 + \varepsilon - \lambda} < a$, where the inequality holds for all the domain of $a$ due to the lower bound of $\varepsilon$ in \cref{eq:bounds_eps}.

    (ii) If $g_1$ and $g_2$ get allocated to agent 2, then the adversary reveals true values $v_1(g_3) = \frac{1}{2} - \varepsilon - \lambda$, $v_1(g_4) = \frac{1}{2} - \varepsilon - \lambda$, and $v_2(g_3) = \frac{1}{2} - 3 \varepsilon - \lambda$, $v_2(g_4) = \frac{1}{2} + \varepsilon - \lambda$. The analysis is symmetric to that of case (i) and shows that all allocations induce an \praefx allocation for $a' < a$.

    (iii) If $g_1$ gets allocated to agent 1 and $g_2$ gets allocated to agent 2, then the adversary reveals true values $v_1(g_3) = \frac{1}{2} - 3 \varepsilon - \lambda$, $v_1(g_4) = \frac{1}{2} + \varepsilon - \lambda$, and $v_2(g_3) = \frac{1}{2} - 3 \varepsilon - \lambda$, $v_2(g_4) = \frac{1}{2} + \varepsilon - \lambda$. 
    If both $g_3$ and $g_4$ get allocated to the same agent, w.l.o.g. agent 2, then agent 1 is EFX-envious towards agent 2, and $v_1(\{g_2, g_3, g_4\}) - \min \{  v_1(g_2), v_1(g_3), v_1(g_4) \} \leq v_1(\{g_2, g_4\}) = \frac{1}{2} + \varepsilon + \lambda$. So this is an \praefx allocation with $a' \leq \frac{2 \varepsilon}{1/2 + \varepsilon + \lambda} < a$, which holds due to the fact that $\varepsilon < \frac{a}{2(2-a)} - \lambda \cdot \frac{2+a}{2-a} \leq \left( \frac{1}{2} + \lambda \right) \frac{a}{2-a}$ for $\lambda \geq 0$, where the first inequality comes from \cref{eq:bounds_eps}. If the two last goods get allocated to distinct agents, w.l.o.g., $g_3$ gets allocated to agent 1 and $g_4$ to agent 2, then agent 1 is EFX-envious towards agent 2, and $v_1(\{g_2, g_4\}) - \min \{  v_1(g_2), v_1(g_4) \} \leq v_1(g_4) = \frac{1}{2} + \varepsilon - \lambda$. This is an \praefx allocation with $a' \leq \frac{1/2 - \varepsilon - \lambda}{1/2 + \varepsilon - \lambda} < a$, where the inequality holds for all the domain of $a$ due to the lower bound of $\varepsilon$ in \cref{eq:bounds_eps}.

    (iv) If $g_1$ gets allocated to agent 2 and $g_2$ gets allocated to agent 1, then the adversary reveals the same true values as in case (iii). Here we simply use the fact that $\varepsilon \geq \lambda$: this is obvious when $a \in \left( \frac{2}{3}, 1 \right]$, and for the case where $a \in \left( \frac{1}{2}, \frac{2}{3} \right]$, we have $\lambda \leq \frac{2a - 1}{6a} \leq \frac{1}{12}$, which implies, $\varepsilon = \frac{1}{6} - \lambda \geq \frac{1}{12} \geq \lambda$. Therefore, agents 1 and 2 now have bundles that are worth to them at most as much as they were worth in case (iii), and so if in case (iii) there was no \aefx allocation, the same holds here.

    Finally, notice that the error between the prediction and the true values in all the above cases is $\varepsilon$: if $a \in \left( \frac{1}{2}, \frac{2}{3} \right]$, then $\varepsilon = \frac{1}{6} - \lambda$, and $\lambda$ can take any value strictly smaller than $\frac{2a-1}{6a}$, so $\varepsilon$ can take any value strictly greater than $\frac{1-a}{6a}$; if $a \in \left( \frac{2}{3}, 1 \right]$, then $\varepsilon = \lambda$, and $\lambda$ can take any value strictly greater than $\frac{1-a}{4}$.
\end{proof}

Figure~\ref{fig:2_agents_general} illustrates the case of two agents with additive, normalized, but not necessarily identical valuations. The dotted green line shows the upper bound on the prediction error (corresponding to the lower bound on the accuracy) proven in \cref{thm:with_predictions_non-id_2_ag_a-EFX_inapprox}, which holds for $a \in (1/2 , 1]$. For some $a \in (0, 1/2]$, it is still possible that there is an algorithm that guarantees an \aefx allocation when the error is bounded away from 1 by a constant -- we have no such algorithm, though. The dash-dot blue line shows the lower bound on the error (corresponding to the upper bound on the accuracy) provided by \cref{cor:with_predictions_n_ag_a-EFX_positive} for $\tila = 1$.\footnote{For the case of two agents with additive valuations, there is a simple offline algorithm (cut-and-choose) that guarantees an exact EFX: first, run LPT (\cref{alg:LPT}) on agent 1's prediction to create a pair of bundles which constitute an exact EFX according to $p_1$; then, let agent 2 choose which bundle she prefers according to $p_2$, and the other one is given to agent 1.}

 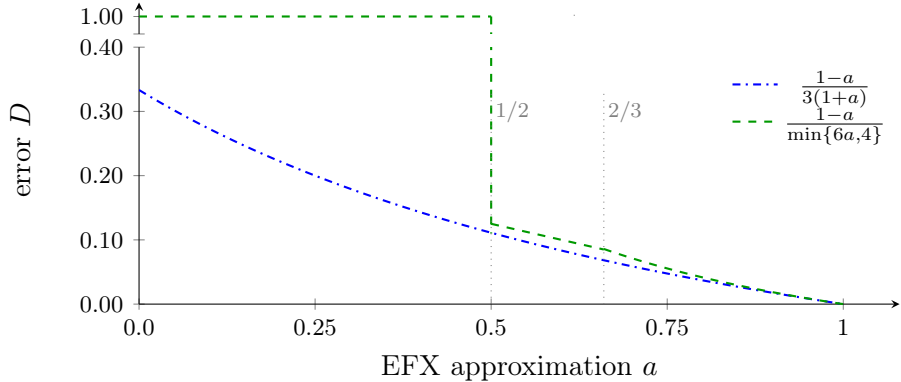
\begin{figure}[tbh]
     \centering
     \begin{tikzpicture}
\begin{groupplot}[
  group style={
    group name=broken,
    group size=1 by 2,
    vertical sep=5pt
  },
  tick label style={font=\footnotesize},
  width=0.75*\textwidth,
  xmin=0.0, xmax=1.08,
  xtick={0, 0.25, 0.5, 0.75, 1},
  xticklabels={0.0, 0.25, 0.5, 0.75, 1},
  domain=0.6:1.1,
  samples=200,
  scaled y ticks=false,
  scaled x ticks=false,
  tick label style={
    /pgf/number format/fixed,
    /pgf/number format/precision=2,
    /pgf/number format/fixed zerofill
  }
]

\nextgroupplot[
  height=2cm,
  ymin=0.9, ymax=1.08,
  ytick={0.9,1},
  yticklabels={,1.00},
  xticklabels={},
  y axis line style={-}, 
  axis lines=left,
  axis x line=none,
]

\addplot[green!60!black, dashed, thick, domain=0:0.5] {1};

\addplot[blue, thick, dash dot, domain=0.0:1] 
  {(1 - x)/(3*(1+x))};

\addplot[green!60!black, thick, dashed] coordinates {(0.5, 1.0) (0.5, 0.16)};

\addplot[dotted, gray] coordinates {(0.618, 1.01) (0.618, 0.98)};

\nextgroupplot[
  height=5cm,
  ymin=0, ymax=0.4,
  axis lines=left,
  y axis line style={-}, 
  ytick={0, 0.1, 0.2, 0.3, 0.4},
  xlabel={EFX approximation $a$},
  ylabel style={at={(axis description cs:-0.13,0.6)}},
  ylabel={error $D$},
  legend style={
    at={(1,0.95)},
    anchor=north east,
    draw=none,
    fill=none,
    font=\small
  }
]

\addlegendimage{line legend, blue, thick, dash dot}
\addlegendentry{$\frac{1 - a}{3(1 + a)}$}

\addlegendimage{line legend, green!60!black, thick, dashed}
\addlegendentry{$\frac{1 - a}{\min\{6 a, 4\}}$}



\addplot[green!60!black, thick, dashed] coordinates {(0.5, 1.0) (0.5, 0.11)};



\addplot[blue, thick, dash dot, domain=0.0:1] 
  {(1 - x)/(3*(1+x))};

\addplot[green!60!black, thick, dashed, domain=0.5:0.66] 
  {(1 - x)/(4)};
\addplot[green!60!black, dashed, thick, domain=0.66:1] 
  {(1 - x)/(6*x)};


\addplot[dotted, gray] coordinates {(0.66, 0) (0.66, 0.33)};
\node[gray] at (axis cs:0.69,0.30) {\scriptsize$2/3$};

\addplot[dotted, gray] coordinates {(0.5, 0) (0.5, 0.33)};
\node[gray] at (axis cs:0.53,0.30) {\scriptsize$1/2$};

\end{groupplot}
\end{tikzpicture}    
     \caption{Prediction error $D=1-\eta$ as a function of $a \in [0, 1]$ for 2 agents with additive, normalized valuations. Dash-dot plot: \cref{cor:with_predictions_n_ag_a-EFX_positive}, dashed plot: \cref{thm:with_predictions_non-id_2_ag_a-EFX_inapprox}.}
     \label{fig:2_agents_general}
 \end{figure}

When the valuation functions are identical, the accuracy lower bounds become slightly better. Recall that in \cref{thm:no_predictions_phi-EFX_upper_bound} we showed a simple greedy algorithm that achieves $(\varphi - 1)$-EFX allocations without predictions, and from \cref{thm:no_predictions_phi-EFX_inapprox_bound} we deduced that to find a better approximation, we need access to predictions. In the following theorem, we show that even if an algorithm has access to predictions of up to a certain accuracy level, it still cannot guarantee an \aefx for any $a \in (\varphi - 1, 1]$.

\begin{theorem}\label{thm:with_predictions_comb-beyond-phi-EFX_inapprox_bound}
    Suppose we have $2$ agents with additive, identical, normalized valuations, with a provided prediction of accuracy $\eta < 1 - \frac{1-a}{\min\{2a(2+a), 4\}}$ for some given $a \in \left( \varphi - 1, 1 \right]$, that is, the error between the prediction and the true valuation is $D = 1 - \eta > \frac{1-a}{\min\{2a(2+a), 4\}}$. Then, there is no algorithm that guarantees an \aefx allocation, even when $T'=T=4$, and the predictions and the true valuations are $4$-value functions.
\end{theorem}

\begin{proof}
    For the sake of contradiction, suppose there is an algorithm that for some $a \in (\varphi - 1, 1]$ and accuracy $\eta < 1 - \frac{1-a}{\min\{2a(2+a), 4\}}$, guarantees an \aefx allocation. 
    
    If $a \in (\varphi - 1, \sqrt{3}-1]$, and therefore, $2a(2+a) \leq 4$, the adversary fixes a rational $r \in \left( 0, D - \frac{1-a}{2a(2+a)} \right]$, sets $\lambda = \max \left\{ \frac{a^2 + a - 1}{2a(2+a)} - r, 0 \right\}$, and an $\varepsilon \in \left( \frac{1-a}{2a(2+a)} + r \cdot \frac{1-a}{1+a}, \frac{1-a}{2a(2+a)} + r \right)$. It then gives the following prediction to the algorithm: $p(g_1) = 2 \lambda$, $p(g_2) = 2 \varepsilon$, $p(g_3) = \frac{1}{2} - 2 \varepsilon - \lambda$, and $p(g_4) = \frac{1}{2} - \lambda$, and observe that $\varepsilon \geq \lambda$. The adversary reveals the true values $v(g_1) =  2 \lambda$, $v(g_2) = 2 \varepsilon$. There are two cases:

    (i) If $g_1$ and $g_2$ get allocated to the same agent, w.l.o.g. agent 1, then the adversary reveals the true values $v(g_3) = v(g_4) = \frac{1}{2} - \varepsilon - \lambda$. If $g_3$ and $g_4$ get allocated to agent 1, this is a $0$-EFX allocation and so the algorithm fails to provide an \aefx allocation for the specified $a$. If $g_3$ and $g_4$ get allocated to agent 2, then agent 2 does not EFX-envy agent 1, since $1 - 2 (\varepsilon + \lambda) \geq 2 \varepsilon$ for the given domains of $\varepsilon$ and $a$, and this is an \praefx allocation with $a' = \frac{2 (\varepsilon + \lambda)}{1/2 - \varepsilon - \lambda} < a$, where the inequality comes from the fact that $\varepsilon < \frac{1-a}{2a(2+a)} + r$. If $g_3$ is allocated to agent 1 and $g_4$ is allocated to agent 2, then agent 1 does not envy agent 2, and notice that $g_1$ is the smallest-valued good that agent 1 has, since $2 \lambda \leq 2 \varepsilon \leq \frac{1}{2} - \varepsilon - \lambda$ for the given domains of $\varepsilon$ and $a$. This is an \praefx allocation with $a' = \frac{1/2 - \varepsilon - \lambda}{1/2 + \varepsilon - \lambda} < a$, where the inequality comes from the fact that $\varepsilon > \frac{1-a}{2a(2+a)} + r \cdot \frac{1-a}{1+a}$. If $g_3$ is allocated to agent 2 and $g_4$ is allocated to agent 1, a symmetric argument applies.

    (ii) If w.l.o.g. $g_1$ is allocated to agent 1 and $g_2$ is allocated to agent 2, then the adversary reveals the true values $v(g_3) = \frac{1}{2} - 3 \varepsilon - \lambda$ and $v(g_4) = \frac{1}{2} + \varepsilon - \lambda$. Notice that $g_1$ is still the smallest-valued good among all, since $2 \lambda \leq 2 \varepsilon \leq \frac{1}{2} - 3 \varepsilon - \lambda \leq \frac{1}{2} + \varepsilon - \lambda$ for the given domains of $\varepsilon$ and $a$. If both $g_3$ and $g_4$ get allocated to agent 1, then he does not EFX-envy agent 2, and this is an \praefx allocation for $a' = \frac{2 \varepsilon}{1 - 2 (\varepsilon + \lambda)} \leq \frac{2( \varepsilon + \lambda)}{1/2  - \varepsilon - \lambda} < a$, where the last inequality comes from one of the bounds in case (i). If both $g_3$ and $g_4$ get allocated to agent 2, then he does not EFX-envy agent 1, and this is an \praefx allocation for $a' = \frac{2 \lambda}{1 - 2 (\varepsilon + \lambda)} \leq \frac{2 \varepsilon}{1 - 2 (\varepsilon + \lambda)} < a$, where the last inequality comes from the previous bound. If $g_3$ gets allocated to agent 2 and $g_4$ gets allocated to agent 1, then agent 1 does not envy agent 2, and this is an \praefx for $a' = \frac{1/2 - \varepsilon - \lambda}{1/2 + \varepsilon - \lambda} < a$, where the inequality comes from one of the bounds in case (i). If $g_3$ goes to agent 1 and $g_4$ goes to agent 2, then agent 2 does not envy agent 1, and this is an \praefx for $a' = \frac{1/2 - 3 \varepsilon + \lambda}{1/2 + \varepsilon - \lambda} \leq \frac{1/2 - \varepsilon - \lambda}{1/2 + \varepsilon - \lambda} < a$, where the last inequality comes from the previous bound.

    From the above cases we can see that the algorithm fails to output an \aefx allocation. Now observe that the error between the prediction and the true values in the two cases is equal to $\varepsilon < \frac{1-a}{2a(2+a)} + r \leq D$, which contradicts our assumption that the algorithm can provide an \aefx allocation for any error of value (at most) $D$.

    If $a \in (\sqrt{3} - 1, 1]$, or equivalently, $2a(2+a) > 4$, the adversary fixes a rational $\varepsilon \in \left( \frac{1-a}{4}, \min \left\{ \frac{a}{4(2+a)}, D \right\} \right)$ and gives the following prediction to the algorithm: $p(g_1) = 2 \varepsilon$, $p(g_2) = 2 \varepsilon$, $p(g_3) = \frac{1}{2} - 3 \varepsilon$, and $p(g_4) = \frac{1}{2} - \varepsilon$. Then, we follow the same case analysis as above, where now we substitute $\lambda$ with $\varepsilon$. Therefore, the algorithm fails to output an \aefx allocation. Notice that, again, the error of this instance equals $\varepsilon$, which can take any value in $\left( \frac{1-a}{4}, \min \left\{ \frac{a}{4(2+a)}, D \right\} \right)$, a non-empty interval for $a \in (\sqrt{3}-1, 1]$. Therefore, the adversary can choose some $\varepsilon$ arbitrarily close to the lower limit of that interval, making the algorithm fail to provide an \aefx for error less than $D$, a contradiction.
\end{proof}

When we have $n \geq 3$ agents, this bound becomes significantly worse, especially for large values of $a$. Furthermore, we get non-zero lower bounds for small values of $a$, contrary to the case of two agents. We show this in two auxiliary lemmas, specifically \cref{lem:with_predictions_3_ag_ternary_a-EFX_inapprox} and \cref{lem:with_predictions_n_ag_ternary_a-EFX_inapprox}, where the former dominates for small values of $a$, and the latter for large values of $a$. \cref{thm:with_predictions_n_ag_ternary_a-EFX_inapprox_comb} combines them into a single statement.

\begin{lemma}\label{lem:with_predictions_3_ag_ternary_a-EFX_inapprox}
    Suppose we have $n \geq 3$ agents with additive, identical, normalized valuations, with a provided prediction of accuracy $\eta < 1 - \frac{1}{2(n-1 + 2a)}$ for some given $a \in \left( 0, 1 \right]$, that is, the error between the prediction and the true valuation is $1 - \eta > \frac{1}{2(n-1 + 2a)}$. Then, there is no algorithm that guarantees an \aefx allocation, even when $T'=T=n+1$, and the predictions and the true valuations are $3$-value functions.
\end{lemma}


\begin{proof}
    For the sake of contradiction, suppose there is an algorithm that can provide an \aefx allocation for some $a \in \left( 0, 1 \right]$. The algorithm knows that $T=n+1$ and has accuracy $\eta < 1 - \frac{1}{2(n-1 + 2a)}$. The adversary fixes a rational value $\varepsilon \in \left( 0, \frac{a}{n-1+2a} \right)$, and gives to the algorithm the prediction: $p(g_1) = p(g_2) = \varepsilon$, $p(g_t) = \frac{2n - 3}{(n-1)(n-2)} \left( \frac{1}{2} - \varepsilon \right)$ for $t \in \{ 3, 4, \dots, n \}$, and $p(g_{n+1}) = \frac{1}{n-1} \left( \frac{1}{2} - \varepsilon \right)$. The adversary reveals the true values of the first two goods to the algorithm, namely $v(g_1) = v(g_2) = \varepsilon$. There are two cases: 
    
    (i) If the algorithm allocates the two goods to the same agent, say $i$, then the adversary provides $n-2$ goods of true value $p(g_t) = \frac{1 - 2 \varepsilon}{n-2}$ for $t \in \{ 3, 4, \dots, n \}$, and a final good of true value 0.\footnote{The latter good of value $0$ is given so that the adversary fulfills the promise that $T=n+1$, and does not affect the bounds on the EFX approximation ratio.} Since the goods of positive value are $n$ and two of them are allocated to the same agent, regardless of where the goods $g_3 , \dots, g_{n+1}$ are allocated, there is an agent with value $0$. Recall now that agent $i$ has two goods, each of value $\varepsilon > 0$, so this is a $0$-EFX allocation, a contradiction to the assumption that $a > 0$. 
    
    (ii) If the algorithm allocates the first two goods to two agents, say $i, j$, then the adversary provides $n-1$ goods of true value $v(g_t) = \frac{1 - 2 \varepsilon}{n-1}$ for $t \in \{ 3, 4, \dots, n+1 \}$. Notice that $\frac{1 - 2 \varepsilon}{n-1} > \varepsilon$ for any $a \leq 1$ by definition of $\varepsilon$, therefore, any of $g_3 , \dots, g_{n+1}$ is more valuable than $g_1$ and $g_2$. By the pigeonhole principle, at least one of the agents will have at least two goods (one of which has value $\frac{1 - 2 \varepsilon}{n-1}$), and at least one of them will have value at most $\varepsilon$. Notice also that, w.l.o.g., agent $i$ will be allocated (at least) one of $\{g_1 , g_2\}$ and one of $\{g_3 , \dots,  g_{n+1} \}$. So, one of the agents is EFX-envious towards $i$ and the allocation is \praefx for $a' \leq \frac{\varepsilon}{(1 - 2 \varepsilon)/(n-1)} < a$, where the last inequality comes from the upper bound of $\varepsilon$. Therefore, in this case too, the algorithm cannot provide an \aefx allocation, a contradiction.

    Finally, the error between the prediction and the two scenarios of true values of the adversary is the following. For case (i) it is 
    \begin{align*}
        \frac{1}{2} \cdot \left( 0 + 0 + (n-2) \cdot \left[ \frac{1 - 2 \varepsilon}{n-2} - \frac{(2n - 3)(1/2 - \varepsilon)}{(n-1)(n-2)} \right] + \left[ \frac{1/2 - \varepsilon}{n-1} - 0 \right] \right) = \frac{1}{n-1} \left( \frac{1}{2} - \varepsilon \right) ,
    \end{align*}
    and for case (ii) it is 
    \begin{align*}
        \frac{1}{2} \cdot \left( 0 + 0 + (n-2) \cdot \left[ \frac{(2n - 3)(1/2 - \varepsilon)}{(n-1)(n-2)} - \frac{1 - 2 \varepsilon}{n-1} \right] + \left[ \frac{1 - 2 \varepsilon}{n-1} - \frac{1/2 - \varepsilon}{n-1} \right] \right) = \frac{1}{n-1} \left( \frac{1}{2} - \varepsilon \right) ,
    \end{align*}
    for any $\varepsilon \in \left( 0, \frac{a}{n-1 + 2a} \right)$ of the adversary's choice.
\end{proof}

\begin{lemma}\label{lem:with_predictions_n_ag_ternary_a-EFX_inapprox}
    Suppose we have $n \geq 3$ agents with additive, identical, normalized valuations, with a provided prediction of accuracy $\eta < 1 - \frac{1 - a^2}{4 + (2n-3)a}$ for some given $a \in \left( 0, 1 \right]$, that is, the error between the prediction and the true valuation is $1 - \eta > \frac{1 - a^2}{4 + (2n-3)a}$. Then, there is no algorithm that guarantees an \aefx allocation, even when $T'=T=2n-1$, and the predictions and the true valuations are $3$-value functions.
\end{lemma}


\begin{proof}
    For the sake of contradiction, suppose there is an algorithm that can provide an \aefx allocation for some $a \in \left( 0, 1 \right]$. The algorithm knows that $T=2n-1$ and has accuracy $\eta < 1 - \frac{1 - a^2}{4 + (2n-3)a}$. The adversary fixes rational values $k \in \left( 0, \frac{a}{4 + (2n-3)a} \right)$ and $\varepsilon \in \left( \frac{1 - a^2}{4 + (2n-3)a}, \frac{1}{4 + (2n-3)a} \right]$, and gives to the algorithm the prediction: $p(g_1) = p(g_2) = \dots = p(g_{2n-3}) = k$, $p(g_{2n-2}) = \frac{1 - (2n-3)k}{2} - \varepsilon$, and $p(g_{2n-1}) = \frac{1 - (2n-3)k}{2} + \varepsilon$. The adversary reveals the true values of the first $2n-3$ goods to the algorithm, namely $v(g_1) = v(g_2) = \dots = v(g_{2n-3}) = k$. For ease of presentation, a good $g_t$ for $t \in [2n-3]$ will be called a ``$k$-good''. After time-step $t = 2n-3$ there are three cases: 

    (i) Exactly one agent, w.l.o.g. agent $n$, does not have any $k$-good. Then, there is at least another agent, w.l.o.g. agent $n-1$, that has at most 1 $k$-good (and therefore, exactly 1 $k$-good); otherwise, at least $n-1$ agents have at least two $k$-goods, so we would have at least $2(n-1) = 2n-2 > 2n-3$ $k$-goods, a contradiction. Also, there is at least another agent, w.l.o.g. agent $n-2$, that has at most 2 $k$-goods; otherwise, we would have at least $n-2$ agents with at least 3 $k$-goods, and agent $n-1$ with exactly 1 $k$-good, therefore $3(n-2) + 1 = 3n - 5 > 2n-3$ $k$-goods in total, a contradiction. Consequently, the adversary reveals the true values $v(g_{2n-2}) = v(g_{2n-1}) = \frac{1 - (2n-3)k}{2}$. 
    
    If neither $g_{2n-2}$ nor $g_{2n-1}$ gets allocated to agent $n$, then this is a $0$-EFX allocation, so the algorithm fails to produce an \aefx allocation for the required value of $a$. If one of these goods gets allocated to agent $n$ and the other gets allocated to an agent with exactly one $k$-good, w.l.o.g. agent $n-1$, then one of the agents with at most $2$ $k$-goods, such as agent $n-2$, is EFX-envious towards agent $n-1$, and due to that, this is an \praefx allocation with 
    \begin{align}\label{eq:bound_n_agents_inapprox}
        a' \leq \frac{2k}{(1 - (2n-3)k)/2} = \frac{4}{1/k - (2n-3)} < \frac{4}{4/a + (2n-3) - (2n-3)} = a ,
    \end{align}
    where the last inequality comes from the fact that $k < \frac{a}{4 + (2n-3)a}$. If one of the goods gets allocated to agent $n$ and the other one gets allocated to an agent with more than one $k$-goods, then agent $n-1$ EFX-envies that agent and this is an \praefx allocation with $a' \leq \frac{k}{(1 - (2n-3)k)/2} \leq \frac{2k}{(1 - (2n-3)k)/2} < a $, where the last inequality comes from \cref{eq:bound_n_agents_inapprox}. If both get allocated to agent $n$, then agent $n-1$ is EFX-envious towards her, and due to that, this is an \praefx allocation with $a' \leq \frac{k}{(1 - (2n-3)k)/2} \leq \frac{2k}{(1 - (2n-3)k)/2} < a$, where again, the last inequality comes from \cref{eq:bound_n_agents_inapprox}. So, the algorithm fails to produce an \aefx allocation.

    (ii) There are two agents, w.l.o.g. agents $n-1$ and $n$, that have no $k$-good. If there are at least three agents with no $k$-good then no matter where $g_{2n-2}, g_{2n-1}$ get allocated, this will be a $0$-EFX allocation; because at least one agent will have no goods, and there will be at least one agent with 2 $k$-goods (which have positive value), since we have more than $n-3$ goods. So the remaining case is when only agents $n-1$ and $n$ have no $k$-good. Then, there is at least one agent, w.l.o.g. agent $n-2$, with at least $3$ $k$-goods; otherwise, $n-2$ agents would have at most $2$ $k$-goods, so the total number of such goods would be at most $2(n-2) = 2n - 4 < 2n-3$, a contradiction. Now the adversary reveals the true values $v(g_{2n-2}) = \frac{1 - (2n-3)k}{2} - 2 \varepsilon$ and $v(g_{2n-1}) = \frac{1 - (2n-3)k}{2} + 2 \varepsilon$.
    If after the algorithm allocates $g_{2n-2}, g_{2n-1}$, some agent among agents $n-1$ and $n$ has not received any good, then this is a $0$-EFX allocation, since another agent exists with at least $2$ $k$-goods. If each of agents $n-1$ and $n$ receives one of the last two goods, w.l.o.g. let agent $n-1$ receive $g_{2n-2}$, and agent $n$ receive $g_{2n-1}$, then the latter EFX-envies agent $n-2$ who has $3$ $k$-goods. This will be an \praefx allocation with $a' \leq \frac{(1 - (2n-3)k)/2 - 2 \varepsilon}{2k} < \frac{4ak/2}{2k} = a$, where the last inequality comes from the fact that $\varepsilon > \frac{1-a^2}{4 + (2n-3)a} > \frac{1-(2n-3 + 4a)k}{4}$. Therefore, the algorithm fails to output an \aefx allocation.
    
    (iii) Every agent has at least one $k$-good. Then, there are at least 3 agents, w.l.o.g. agents $n-2$, $n-1$, and $n$, that have exactly one $k$-good; otherwise, at least $n-2$ agents would have at least 2 such goods each, and 2 agents would have exactly 1 $k$-good, so the total number of $k$-goods would be at least $2(n-2) + 2 = 2n-2 > 2n-3$, a contradiction. The adversary reveals the true values $v(g_{2n-2}) = \frac{1 - (2n-3)k}{2} - 2 \varepsilon$ and $v(g_{2n-1}) = \frac{1 - (2n-3)k}{2} + 2 \varepsilon$. No matter where these last two goods get allocated, the agent that will receive $g_{2n-1}$ will also have at least one $k$-good, and at least one agent, w.l.o.g. agent $n$, will only have a single $k$-good. Therefore, the latter agent will EFX-envy the former and this will be an \praefx allocation for $a' \leq \frac{k}{(1 - (2n-3)k)/2 + 2 \varepsilon} \leq \frac{2k}{(1 - (2n-3)k)/2} < a$, where the last inequality is due to \cref{eq:bound_n_agents_inapprox}. So, the algorithm cannot provide an \aefx allocation. 

    Finally, notice that the error between the prediction and the true valuation in this instance equals $\varepsilon$, which can take any value in $\left( \frac{1 - a^2}{4 + (2n-3)a}, \frac{1}{4 + (2n-3)a} \right]$.
\end{proof}

\begin{theorem}\label{thm:with_predictions_n_ag_ternary_a-EFX_inapprox_comb}
    Suppose we have $n \geq 3$ agents with additive, identical, normalized valuations, with a provided prediction of accuracy $\eta < 1 - \min \left\{\frac{1}{2(n-1 + 2a)}, \frac{1 - a^2}{4 + (2n-3)a} \right\}$ for some given $a \in \left( 0, 1 \right]$, that is, the error between the prediction and the true valuation is $1 - \eta > \min \left\{\frac{1}{2(n-1 + 2a)}, \frac{1 - a^2}{4 + (2n-3)a} \right\}$. Then, there is no algorithm that guarantees an \aefx allocation, even when $T'=T=2n-1$, and the predictions and the true valuations are $3$-value functions.
\end{theorem}

\begin{proof}
    If $\frac{1}{2(n-1 + 2a)} \leq \frac{1 - a^2}{4 + (2n-3)a}$, then the adversary provides the algorithm with the instance specified in the proof of \cref{lem:with_predictions_3_ag_ternary_a-EFX_inapprox}, and produces $n-2$ dummy goods $g_{n+2}, g_{n+3}, \dots, g_{2n-1}$ of value 0 to fulfil the promise that $T = 2n-1$.\footnote{Notice that these extra goods do not affect the analysis since they can be allocated to an agent that is not envied.} If $\frac{1}{2(n-1 + 2a)} > \frac{1 - a^2}{4 + (2n-3)a}$, then the adversary provides the algorithm with the instance specified in the proof of \cref{lem:with_predictions_n_ag_ternary_a-EFX_inapprox}.
\end{proof}

One can observe from \cref{thm:with_predictions_n_ag_ternary_a-EFX_inapprox_comb} that when $n \geq 3$, for any fixed $a \in (0, 1)$, the lower bound on the accuracy is a function $1 - \Theta (1/n)$. The upper bound on the accuracy given in \cref{cor:with_predictions_n_ag_a-EFX_positive} for identical valuations (and therefore $\tila = 1$ by employing \cref{alg:LPT}) becomes $1 - \frac{1-a}{(2n-1)(1+a)}$, which is also a $1 - \Theta (1/n)$ function. This means that as the number of agents grows, the necessary accuracy level tends to $100 \%$, and is asymptotically matched by the accuracy that suffices for a simple algorithm relying entirely on predictions. This fact at first glance might seem to make the case that for large $n$, finding better algorithms than those that blindly use predictions is not interesting. However, for fixed $a$ and $n$, there is still a non-negligible gap between the lower and upper bounds of accuracy, which is worth investigating. In the following section, we try to bridge the largest bound-gap for identical valuations, which is for the case of two agents.

\subsubsection{Positive results}

We focus on two agents with additive, identical valuations, and present \cref{alg:id_2_ag} which, for the same approximation factor $a$, requires less accuracy than that of \cref{thm:with_predictions_n_ag_a_i-EFX_positive} (or equivalently, for the same accuracy, it provides an approximate EFX allocation with better approximation factor than that of the aforementioned result). This \aefx algorithm has a known accuracy that is a function $\eta(a)$ of $a \in [0,1]$. It starts by finding an exact EFX allocation $A$ using LPT (\cref{alg:LPT}) on the predicted values, and then uses this as a guideline for the final allocation. According to the form that this allocation has, i.e., its characteristics, it prescribes appropriate thresholds for true values of the goods that arrive over time, and allocates them to the right agent so that the final allocation is \aefx.

\begin{theorem}\label{thm:with_predictions_2_ag_ternary_a-EFX_positive}
    Suppose we have $2$ agents with additive, identical, normalized valuations, with a provided prediction of accuracy $\eta \geq 1 - \dmax$ for some given $a \in \left( \varphi - 1, 1 \right]$, that is, the error between the prediction and the true valuation is $1 - \eta \leq \dmax$. Then, \cref{alg:id_2_ag} outputs an \aefx allocation, and performs a constant number of basic operations per time-step.
\end{theorem}

\begin{proof}
    \cref{alg:id_2_ag} first uses LPT (\cref{alg:LPT}) to find an exact EFX allocation, called $A$. Then, $A$ is classified according to its \emph{form}, that is, some important properties of the two bundles $A_1, A_2$. Finally, a form-specific routine is used, which, at each time-step decides the recipient of the good in constant time, resulting in an \aefx allocation $B$.

    Let us denote the maximum error between the prediction and the true valuation by $\dm := \dmax$, for $a \in (\varphi - 1, 1]$. In the trivial case where $a=1$, $\dm = 0$, the true values are identical to the predicted values (with the exception of additional goods that might arrive with true value 0 which do not affect the approximation factor the solution), therefore \cref{alg:id_2_ag} outputs an $1$-EFX allocation $A$. In what follows, we will be considering $a \in (\varphi-1, 1)$. W.l.o.g., let agent 1 have at most as much value as agent 2 in $A$, i.e., $p(A_1) \leq p(A_2)$, and let us denote by $g$ an arbitrary good in $\arg \min_{g' \in A_2} p(g')$. First, observe that, since $A$ has been produced by LPT, $p(A_1) \geq p(A_2 \setminus g)$. Therefore, $p(A_1) \geq 1/3$, otherwise, $p(A_1) < 1/3$ implies that $p(A_2 / g) \leq p(A_1) < 1/3$, and so, $p(g) \leq p(A_2 \setminus g) < 1/3 $. This means that $p(A_1) + p(A_2 \setminus g) + p(g) < 1$, which contradicts the normalization condition. 
    
    The following is a preliminary result for the case where the prediction consists of at most three goods, and will be useful in our main algorithm (\cref{alg:id_2_ag}).

    \begin{lemma}\label{lem:3_goods}
        Suppose we have $2$ agents with additive, identical, normalized valuations, with a provided predicted horizon $T' \leq 3$ and accuracy $\eta \geq 1 - \dmaxthree$ for some given $a \in [ 0, 1 ]$, that is, the error between the prediction and the true valuation is $1 - \eta \leq \dmaxthree$. Then, \cref{alg:3_goods} guarantees an \aefx allocation. Furthermore, if the algorithm knows that the predicted horizon is the true one (i.e., $T = T'$), then it guarantees an exact EFX even for prediction accuracy 0 (or equivalently, error $1$).
    \end{lemma}
    

    \begin{proof}
        Suppose we have an instance and algorithm as described in the first part of the statement, i.e., a predicted horizon $T' \leq 3$ and accuracy $\eta \geq 1 - \dmaxthree$ for some given $a \in [ 0, 1 ]$. Interestingly, the algorithm does not require a prediction vector. Essentially, it does the following. It starts by placing the first good ($g_1$) to agent 1. If there is no other round, then this is an exact EFX allocation; if there is another round, $t=2$, the good $g_2$ arrives and the algorithm now assumes that a third one ($g_3$) will have the entire remaining value of $1 - v(g_1) - v(g_2)$. It then calculates which one has the greatest true value by comparing $v(g_1), v(g_2)$ and $1 - v(g_1) - v(g_2)$.
        \begin{itemize}
            \item If $\max \{ v(g_1), v(g_2), 1 - v(g_1) - v(g_2) \}$ equals $v(g_1)$ or $v(g_2)$, then $g_2$ is allocated to $A_2$ so that the highest-valued good remains alone in its bundle. If $v(g_2) = 1 - v(g_1)$, then even if the true horizon is $T \geq 3$, $v(g_t) = 0$ for all $t \in \{ 3, 4, \dots, T \}$, and all goods $g_t$ will be placed to the agent with lowest value between $v(g_1), v(g_2)$, which will be an exact EFX allocation. 
            \item If $\max \{ v(g_1), v(g_2), 1 - v(g_1) - v(g_2) \} = 1 - v(g_1) - v(g_2)$, then this means that $v(g_2) < 1 - v(g_1)$ and, as mentioned earlier, the algorithm assumes that $g_3$ has value $1 - v(g_1) - v(g_2)$. Then it allocates $g_2$ to $A_1$ (which now joins $g_1$), and at $t=3$ it allocates $g_3$ to $A_2$.    
        \end{itemize} 

        In the latter case, if at $t=3$ we have $v(g_3) = 1 - v(g_1) - v(g_2)$, then any goods $g_t$ for $t \in \{ 4, 5, \dots, T \}$ that will potentially arrive have $v(g_t) = 0$ and will be placed to the least valued bundle, creating an exact EFX allocation. Up to this point, we have proven the second part of the statement. However, notice that it might be the case that at $t=3$, we have $v(g_3) < 1 - v(g_1) - v(g_2)$, or equivalently, $v(g_1) + v(g_2) + v(g_3) < 1$, which means that among the aforementioned goods $g_t$, for $t \in \{ 4, 5, \dots, T \}$, there are some with positive value. The algorithm places all goods $g_t$ to agent $i \in [2]$ that has the least valued bundle after $t=3$, that is, $A_{i}^{3}$. In other words, $v(A_{i}^{T'}) \leq v(A_{3-i}^{T'})$. There are two cases:
        \begin{enumerate}
            \item [(i)] $v(A_{i}^T) \leq v(A_{3-i}^T)$. Therefore, agent $3-i$ does not envy agent $i$. Also, $v(A_{i}^T) = v(A_{i}^{T'}) + 1 - v(g_1) - v(g_2) - v(g_3) \geq v(\xset{A}_{3-i}^{T'}) = v(\xset{A}_{3-i}^{T})$, where the first equality comes from the fact that all goods $g_t$, for $t \in \{ 4, 5, \dots, T \}$, will be allocated to $A_{i}^{T'}$,  and their cumulative value is $1 - v(g_1) - v(g_2) - v(g_3)$; the first inequality comes from the fact that $g_1$ and $g_2$ where placed by the algorithm such that if $v(g_3) = 1 - v(g_1) - v(g_2)$ then $v(A_{i}^{T'}) \geq v(\xset{A}_{3-i}^{T'})$, but as discussed above, $v(g_3)$ can have at most some discrepancy $1 - v(g_1) - v(g_2) - v(g_3)$ which remains to come from goods $g_t, t \in \{ 4, 5, \dots, T \}$; however, notice that in case $g_3$ is placed at $A_{i}^{T'}$, the aforementioned goods $g_t$ will join that bundle, restoring the remaining value, i.e., $v(A_{i}^{T'}) + 1 - v(g_1) - v(g_2) - v(g_3) \geq v(\xset{A}_{3-i}^{T'})$; finally, the last equality comes from the fact that $A_{3-i}$ received no more goods after $t = T'$. Therefore, this is an exact EFX allocation.
            \item [(ii)] $v(A_{i}^T) > v(A_{3-i}^T)$. Therefore, agent $i$ does not envy agent $3-i$. As mentioned in the previous case, $g_3$ might have some discrepancy $1 - v(g_1) - v(g_2) - v(g_3)$ in its value. By the statement's assumption, this discrepancy is upper bounded by \dmaxthree, therefore we have $v(A_{i}^{T'}) + v(A_{3-i}^{T'}) = v(g_1) + v(g_2) + v(g_3) \geq 1 - \dmaxthree$. Since $v(A_{i}^{T'}) \leq v(A_{3-i}^{T'})$, we have $v(A_{3-i}^{T'}) \geq \frac{1}{2} - \frac{1}{2} \cdot \frac{1-a}{1+a}$. Therefore, this is an \praefx allocation with $a' = \frac{v(A_{3-i}^{T})}{v(\xset{A}_{i}^{T})} \geq \frac{v(A_{3-i}^{T'})}{v(A_{i}^{T'}) + (1-a)/(1+a)} \geq \frac{1/2 - (1-a)/(2(1+a))}{1/2 + (1-a)/(2(1+a))} = a$, where the first inequality comes from the fact that $v(\xset{A}_{i}^T) \leq v(A_{i}^T)$, and $v(A_{i}^T)$ is at most $v(A_{i}^{T'})$ plus extra value $1 - v(g_1) - v(g_2) - v(g_3) \leq \dmaxthree$ from the goods $g_t$, $t \in \{ 4, 5, \dots, T \}$ received.
        \end{enumerate}
        
        Finally, notice that the above arguments did not use any predicted values for the goods, only the prediction that the total value to arrive will be contained within the first three goods.
    \end{proof}

\begin{algorithm}
\caption{Two agents with additive, identical valuations, and predicted horizon $T' \leq 3$: Computing an \aefx for $a \in [0, 1]$ when prediction accuracy is at least $1 - \dmaxthree$}\label{alg:3_goods}  
	\begin{algorithmic}[1]
		\REQUIRE{A value $T' \in [3]$ of predicted horizon.}
		\ENSURE{An \aefx allocation $B$.}
		
		\medskip
		
		\STATE{$(B_1, B_2) \gets (\{ g_1 \}, \emptyset)$}
            \COMMENT{time-step $t = 1$: $g_1$ is allocated to an agent, w.l.o.g., agent 1} 
        
        \medskip
        
        \COMMENT{time-step $t = 2$}
        \IF{$\min \{ T', T \} \geq 2$ \AND  $\max \{ v(g_1), v(g_2), 1 - v(g_1) - v(g_2) \} = \max \{ v(g_1), v(g_2) \}$}
            \STATE{$B_2 \gets B_2 \cup \{ g_2 \}$}

            \medskip
            
            \COMMENT{time-step $t = 3$}
            \IF{$\min \{ T', T \} \geq 3$ \AND $\max \{ v(g_1), v(g_2) \} = v(g_1)$}
                \STATE{$B_2 \gets B_2 \cup \{ g_3 \}$}
            \ENDIF
            \IF{$\min \{ T', T \} \geq 3$ \AND $\max \{ v(g_1), v(g_2) \} = v(g_2)$}
                \STATE{$B_1 \gets B_1 \cup \{ g_3 \}$}
            \ENDIF
        \ENDIF

        \medskip
            
        \COMMENT{time-step $t = T'+1, T'+2, \dots, T $, in case $T > T'$}
        \FOR{$t \in \{ T' + 1, T' + 2, \dots, T \}$}
            \IF{$v(B_1^{T'}) \leq v(B_2^{T'})$}
                \STATE{$B_1 \gets B_1 \cup \{ g_t \}$}
            \ELSE
                \STATE{$B_2 \gets B_2 \cup \{ g_t \}$}
            \ENDIF
        \ENDFOR       
	\end{algorithmic}
\end{algorithm}

    \begin{corollary}\label{cor:3_goods_last_small_goods}
        Suppose we have $2$ agents with additive, identical, normalized valuations, with a provided prediction such that $\sum_{t=4}^{T'} p(g_t) \leq \dmaxthree - D$ and accuracy $\eta \geq 1 - D$ for some given $a \in [ 0, 1 ]$ and $D \in \left[ 0, \dmaxthree \right]$, that is, the error between the prediction and the true valuation is $1 - \eta \leq D$. Then, \cref{alg:3_goods} guarantees an \aefx allocation.
    \end{corollary}

    \begin{proof}
        From the proof of \cref{lem:3_goods}, we can see that \cref{alg:3_goods} does not need any predicted values to allocate the first three goods, and what prevents the approximation of the EFX allocation from being 1 is the discrepancy that (virtual) values $p(g_4) = p(g_5) = \dots = p(g_T) = 0$ have with respect to the true ones $v(g_4), v(g_5), \dots, v(g_T)$. In particular, if the error of the instance ($\sum_{t=4}^{T} v(g_t)$) is at most $\dmaxthree$, then the algorithm is guaranteed to provide an \aefx allocation. Given that the prediction accuracy is at least $1 - D$, we know that $\sum_{t=4}^{T} v(g_t) \leq \sum_{t=4}^{T'} p(g_t) + D$, therefore the aforementioned error can be guaranteed if $\sum_{t=4}^{T'} p(g_t) \leq \dmaxthree - D$.  
    \end{proof}

    Next, we show the following.

\begin{proposition}\label{prop:p(A_1)_large_enough}
    If $p(A_1) \geq \pUB = \pAoneUB$, for some $a \in (\varphi - 1, 1)$, then $A$ is an \aefx allocation according to the true values. 
\end{proposition}


\begin{proof}
    By definition of the naming of the agents, $p(A_1) \leq 1/2$, $p(A_2) \geq 1/2$. Suppose that $p(A_1) \geq \pUB$. Then $p(A_2) \leq 1 - \pUB = \frac{2(3 + a)}{(2 + a)(5 - a)}$. There are two cases:
    
    (i) If $v(A_1) \geq v(A_2)$, then agent 1 does not envy agent 2, and $v(A_1) \leq p(A_1) + \dm \leq 1/2 + \dm$, while $v(A_2) \geq p(A_2) - \dm \geq 1/2 - \dm$. So this is an \praefx allocation for $a' \geq \frac{v(A_2)}{v(A_1)} \geq \frac{1/2 - \dm}{1/2 + \dm} \geq a$, by definition of $\dm$.

    (ii) If $v(A_1) < v(A_2)$, then agent 2 does not envy agent 1, and $v(A_1) \geq p(A_1) - \dm$, while $v(\xset{A}_2) \leq p(\xset{A}_2) + \dm \leq p(A_1) + \dm$. Therefore, this is an \praefx allocation for $a' = \frac{v(A_1)}{v(\xset{A}_2)} \geq \frac{p(A_1) - \dm}{p(A_1) + \dm} \geq \frac{\dm(1+a)/(1-a) - \dm}{\dm(1+a)/(1-a) + \dm} = \frac{2a}{2} = a$, where the last inequality comes from the lower bound of $p(A_1)$.

    Furthermore, if the predicted horizon is shorter than the true one, the extra goods $g_t$, $t \in \{ T'+1, T'+2, \dots, T \}$ can be allocated to agent 1, in other words, $A_1 \gets A_1 \cup \{g_{T'+1}, \dots, g_T \}$ and the final allocation remains \aefx. Indeed, case (i) goes through; in case (ii), since agent 1's is the envious agent, the only complication would be if agent 2, received a small (even 0-valued) good $g_t$, $t \in \{ T'+1, T'+2, \dots, T \}$, which would make $v(\xset{A}_2)$ larger than before. However, notice that agent 2 does not receive any such good, so the analysis of this case holds too.

    Notice that the proof goes through for any $\dm \leq \frac{1}{2} \cdot \frac{1-a}{1+a}$.
\end{proof}

So in case $p(A_1) \geq \pUB = \pAoneUB$, allocation $A$ itself is an \aefx. What remains to show is how the algorithm produces an \aefx allocation $B$ when $p(A_1) \in \left[ \frac{1}{3}, \pAoneUB \right)$.

\begin{proposition}\label{prop:bound_card_A2}
    If $p(A_1) < \pUB = \pAoneUB$ for some $a \in (\varphi - 1, 1)$, then $|A_2| \leq 2$. 
\end{proposition}


\begin{proof}
    Let $g \in \arg \min_{g' \in A_2} p(g')$. We have $p(A_2 \setminus g) \leq p(A_1) < \pAoneUB < \frac{4 + 3 \sqrt{5}}{29}$, for the given domain of $a$. This implies that $p(g) = 1 - p(A_1) - p(A_2 \setminus g) \geq 1 - 2\cdot p(A_1)$. For the sake of contradiction, suppose $|A_2| \geq 3$. Then $1 - p(A_1) = p(A_2) \geq 3\cdot p(g) \geq 3 \cdot (1 - 2\cdot p(A_1))$ implies that $p(A_1) \geq \frac{2}{5} > \frac{4 + 3 \sqrt{5}}{29}$, a contradiction.

    Notice that the proof goes through for any $\dm \leq \frac{2}{5} \frac{1-a}{1+a}$.
\end{proof}

It is immediate that if $|A_2| = 0$, then $p(A_1) = 0$ (by definition of the naming of the two agents), which violates the normalization condition. If $|A_2| = 1$, then agent 1 does not EFX-envy agent 2, and so, $A$ remains an \praefx allocation for the true values with $a' \geq \frac{p(A_2) - \dm}{p(A_1) + \dm} \geq \frac{1/2 - \dm}{1/2 + \dm} \geq a$, for all $a \in (\varphi - 1, 1)$. 

The remaining of this proof shows how \cref{alg:id_2_ag} produces an \aefx allocation $B$ given the LPT allocation $A$, where $p(A_1) \in \left[ \frac{1}{3}, \pAoneUB \right)$, $|A_2| = 2$, and $|A_1| \geq 1$ (since otherwise $p(A_1) = 0 < \frac{1}{3}$, a contradiction). For ease of presentation, we will categorize the values of a given prediction $p$ in three groups: $G^z$, containing all goods with highest value $z \in [0 , 1]$ in $p$, $G^y$, containing all goods with the second-highest value $y \in [0, z)$, and $G^x$, containing the rest, with values strictly smaller than $y$. We will also slightly abuse the notation and say that a good is of type $g^z$, $g^y$, and $g^x$, if it belongs to $G^z$, $G^y$, and $G^x$, respectively. Having satisfied the constraint that $A$ is the outcome of the LPT algorithm (\cref{alg:LPT}) and that $|A_2| = 2$, we can have the following forms of $A_2$:

\begin{enumerate}
    \item[\underline{Form 1:}] $A_2 = \{ g_1^z, g_2^z \}$. Then, $A_1$ contains at least one $g^z$ good, otherwise LPT would not have given both $g_1^z, g_2^z$ to $A_2$. If it contains two $g^z$ goods, then $p(A_1) \geq p(A_2) \geq \frac{1}{2} > \pUB$, which is a contradiction. Therefore, $A_1 = \{ g_3^z, g_1^y, g_2^y, \dots, g_\ell^y, g_1^x, g_2^x, \dots, g_k^x \}$ for some $\ell, k \geq 0$. There are two cases. Case (i): $v(B_1) \geq v(B_2)$; then agent 1 does not envy agent 2, and $B$ is an \praefx allocation for $a' \geq \frac{v(B_2)}{v(B_1)} \geq \frac{p(A_2) - \dm}{p(A_1) + \dm} \geq \frac{1/2 - \dm}{1/2 + \dm} \geq a$. Case (ii): $v(B_1) < v(B_2)$; then, agent 2 does not envy agent 1, and we have $v(B_1) \geq p(A_1) - \dm = 1 - 2z - \dm$, and $v(\xset{B}_2) \leq z + \frac{\dm}{2}$ (since both goods in $B_2$ have value at most $z + \frac{\dm}{2}$). Therefore, $B$ is an \praefx allocation with $a' \geq \frac{1 - 2z - \dm}{z + \dm / 2} \geq \frac{1/3 - \dm}{1/3 + \dm / 2} \geq a$, for all $a \in (\varphi - 1, 1)$, where the second to last inequality comes from the fact that $z \leq 1/3$ (otherwise, $p(A_1) + p(A_2) > 1$, a contradiction). 

    Notice that the above inequalities hold even if $\dm = \dmaxformone$, which is larger than $\dmax$ for all $a \in (\varphi - 1, 1)$.
    
    \item[\underline{Form 2:}] $A_2 = \{ g_1^y, g_2^y \}$. Then, $A_1$ has to contain at least one $g^z$ good ($g^y$ is only defined if a $g^z$ exists). Also, $A_1$ contains at most one $g^z$ good, otherwise LPT would have given one $g^z$ good to $A_2$. Finally, $A_1$ cannot contain both a $g^z$ and a $g^y$ good, otherwise $p(A_1) > p(A_2)$, a contradiction. So, it must be $A_1 = \{ g^z, g_1^x, g_2^x, \dots, g_k^x \}$ for some $k \geq 0$. There are two cases. Case (i): $v(B_1) \geq v(B_2)$; then, agent 1 does not envy agent 2, and $B$ is an \praefx allocation for $a' \geq \frac{v(B_2)}{v(B_1)} \geq \frac{p(A_2) - \dm}{p(A_1) + \dm} \geq \frac{1/2 - \dm}{1/2 + \dm} \geq a$. Case (ii): $v(B_1) < v(B_2)$; then, agent 2 does not envy agent 1, and among the goods $g_1^y, g_2^y, g^z$ at most one's true value can exceed its predicted value by more than $\frac{\dm}{2}$ (otherwise the error between the prediction and the true valuation is greater than $\dm$, a contradiction). Also, according to the algorithm, $B_1$ can only have one of these goods. 
    
    There are two subcases. 
    \begin{enumerate}
        \item [(I)] $B_1$ contains a $g^y$ good, w.l.o.g., $g_1^y$. If $v(g_1^y) \leq y + \frac{\dm}{2}$, then this means that both goods of $B_2$, namely $g_2^y, g^z$, have true value at most $y + \frac{\dm}{2}$; otherwise, one of them would have already been allocated to $B_1$, and $g_1^y$ would be allocated to $B_2$. Then, $v(B_1) \geq v(g_1^y) \geq y - \dm$, and $v(\xset{B}_2) \leq y + \frac{\dm}{2}$. So, $B$ is an \praefx allocation with $a' \geq \frac{y - \dm}{y + \dm/2} \geq a$, where the last inequality comes from the fact that $\dm \leq \frac{1-a}{3+2a}$ for all $a \in (\varphi-1, 1)$, and $y > \frac{1}{2} - \frac{1+a}{2(1-a)} \dm$ since $1-2y = p(A_1) < \frac{1+a}{1-a} \dm$. If $v(g_1^y) > y + \frac{\dm}{2}$, then each of the other two goods' true values must be $v(g_2^y) \leq y + \frac{\dm}{2}$ and $v(g^z) \leq z + \frac{\dm}{2}$ (otherwise, the error between the prediction and the true valuation is greater than $\dm$, a contradiction). Notice also that $z \leq p(A_1) = 1 - 2y$. Therefore, $v(B_1) \geq v(g_1^y) > y + \frac{\dm}{2}$, $v(\xset{B}_2) \leq z + \frac{\dm}{2} \leq 1 - 2y + \frac{\dm}{2}$, and this is an \praefx allocation with $a' \geq \frac{y + \dm/2}{1 - 2y + \dm/2} \geq a$, where the last inequality holds due to the aforementioned lower bound of $y$ and the fact that $\dm \leq \frac{1-a}{a(5+a)}$ for all $a \in (\varphi-1, 1)$. 

        \item [(II)] $B_1$ contains the $g^z$ good. If $v(g^z) \leq y + \frac{\dm}{2}$, then both $g_1^y, g_2^y$ have true value at most $y + \frac{\dm}{2}$ (otherwise, the algorithm would have allocated one of them to $B_1$, and $g^z$ would be placed at $B_2$). Then, $v(B_1) \geq v(g^z) \geq z - \dm \geq y - \dm$, $v(\xset{B}_2) \leq y + \frac{\dm}{2}$, and so, this is an \praefx allocation with $a' \geq \frac{y - \dm}{y + \dm/2} \geq a$, where the last inequality was shown to hold in the previous subcase. If $v(g^z) > y + \frac{\dm}{2}$, then $v(B_1) \geq v(g^z) - \dm \geq y - \frac{\dm}{2}$, and $v(\xset{B}_2) \leq y + \dm$. Therefore, this is an \praefx allocation with $a' \geq \frac{y - \dm/2}{y + \dm} \geq \frac{y - \dm}{y + \dm/2} \geq a$.
    \end{enumerate}

    Finally, notice that the above inequalities would hold even if $\dm = \frac{1-a}{3 + 2a}$, which is larger than $\dmax$ for all $a \in (\varphi-1, 1)$.

    \item[\underline{Form 3:}] $A_2 = \{ g_1^z, g^y \}$. Then $A_1$ contains at most one $g^z$ good, otherwise $p(A_1) > p(A_2)$, a contradiction. If it does not contain any $g^z$ good, then $A_1 = \{ g_1^y, g_2^y, \dots, g_\ell^y, g_1^x, g_2^x, \dots, g_k^x \}$ for some $\ell, k \geq 0$; then it must be $\ell \geq 2$, otherwise LPT would not give $g^y$ to $A_2$. Also, according to the LPT algorithm, $\ell y \geq z$. Furthermore, we know that $\ell y \leq p(A_1) < \pAoneUB$, and $z + y = p(A_2) > 1 - \pAoneUB$. All these, imply that $y \geq 1 - 2 \cdot \pAoneUB = \frac{2 + a + a^2}{(2 + a)(5 - a)}$. Therefore, $\ell \leq \frac{p(A_1)}{y} < \frac{4 + a - a^2}{2 + a + a^2} < 2$, for all $a \in (\varphi - 1, 1)$, a contradiction. So $A_1$ contains a $g^z$ good. Now notice that it cannot contain both a $g^z$ and a $g^y$ good, otherwise $p(A_1) \geq p(A_2) \geq \frac{1}{2} > \pAoneUB$, a contradiction. Therefore, it must be $A_1 = \{ g_2^z, g_1^x, g_2^x, \dots, g_k^x \}$ for some $k \geq 0$. 
    
    There are two cases, and notice that according to \cref{alg:id_2_ag}, all $g^x$ goods are given to agent 1. Case (i): $v(B_1) \geq v(B_2)$; then, agent 1 does not envy agent 2, and $B$ is an \praefx allocation for $a' \geq \frac{v(B_2)}{v(B_1)} \geq \frac{p(A_2) - \dm}{p(A_1) + \dm} \geq \frac{1/2 - \dm}{1/2 + \dm} \geq a$. Case (ii): $v(B_1) < v(B_2)$; which splits into three subcases. 
    \begin{enumerate}
        \item [(I)] The three largest goods of the instance come in the order $g^y, g^z, g^z$. Then, if $v(g^y) \leq z + \dm/2$, it is given to agent 2, and afterwards, one of the $g^z$ goods with true value at most $z + \dm/2$ is given to her too (at most one of $g^y, g^z, g^z$ can have value greater than $z + \dm/2$, otherwise the total error between predictions and true values exceeds $\dm$), while the other $g^z$ good is given to agent 1. Therefore, in this case we get $v(B_1) \geq z + \sum_{j \in [k]}p(g_j^x) - \dm$, and $v(\xset{B}_2) \leq z + \dm/2$. If $v(g^y) > z + \dm/2$, it is given to agent 1, and goods $g_1^z, g_2^z$ are given to agent 2 (and as noted above, both must have true value at most $z + \dm/2$). Therefore in this case we have $v(B_1) > (z + \dm/2) + \sum_{j \in [k]}p(g_j^x) - \dm = z + \sum_{j \in [k]}p(g_j^x) - \dm/2$ , and $v(B_2) \leq z + \dm/2$. So, $B$ is an \praefx allocation with $a' \geq \frac{z + \sum_{j \in [k]}p(g_j^x) - \dm}{z + \dm/2}$. If $z \geq 1/3$, then $a' \geq \frac{1 - \dm/z}{1 + \dm/(2z)} \geq \frac{1/3 - \dm}{1/3 + \dm/2} \geq a$, while if $z < 1/3$, note that $z + \sum_{j \in [k]}p(g_j^x) \geq y + \sum_{j \in [k]}p(g_j^x) = 1 - 2z$, so $a' \geq \frac{1 - 2z - \dm}{z + \dm/2} \geq \frac{1/3 - \dm}{1/3 + \dm/2} \geq a$.
    
        \item [(II)] The three largest goods of the instance come in the order $g^z, g^y, g^z$. If for the first $g^z$ good we have $v(g^z) \leq z + \dm/2$, it goes to agent 2. Afterwards, if $v(g^y) \leq z + \dm/2$, it will go to agent 2 and the remaining $g^z$ will go to agent 1, and if $v(g^y) > z + \dm/2$, the opposite allocation of these two goods will happen. If for the first $g^z$ good we have $v(g^z) > z + \dm/2$, it goes to agent 1 and the remaining $g^y, g^z$ goods go to agent 2, while both have true value at most $z + \dm/2$ as argued earlier. So from all the above cases, we will have $v(B_1) \geq z + \sum_{j \in [k]}p(g_j^x) - \dm$ and $v(\xset{B}_2) \leq z + \dm/2$. Therefore, $B$ is an \praefx allocation with $a' \geq \frac{z + \sum_{j \in [k]}p(g_j^x) - \dm}{z + \dm/2} \geq a$ for all $a \in (\varphi - 1, 1)$, where the last inequality comes from the same analysis as that of the final step in Subcase (I).
    
        \item [(III)] The three largest goods of the instance come in the order $g^z, g^z, g^y$. If for the first $g^z$ good we have $v(g^z) \leq z + \thres$, it goes to agent 2. Then, if the second $g^z$ good has true value $v(g^z) \leq z + \thres$, it goes to agent 2, and the remaining $g^y$ good goes to agent 1, in which case we have $v(B_1) = 1 - v(B_2) \geq 1 - 2 \left( z + \thres \right) = 1 - 2z - \frac{2(1 - a)^2}{(2 + a)(5 - a)}$, and $v(\xset{B}_2) \leq z + \thres$. So, in this case, $B$ is an \praefx allocation with $a' \geq \frac{1 - 2z - 2(1 - a)^2 / ((2 + a)(5 - a))}{z + (1 - a)^2 / ((2 + a)(5 - a))}  \geq a$ for all $a \in (\varphi - 1, 1)$, where the last inequality comes from the fact that $z \leq p(A_1) < \pAoneUB$. If the second $g^z$ good has true value $v(g^z) > z + \thres$ then it goes to agent 1, and the remaining $g^y$ good goes to agent 2, which implies $v(B_1) \geq z + \thres - \dm = z - \incr$, and $v(\xset{B}_2) \leq z + \left( \dm - \thres \right) = z + \incr$. If for the first $g^z$ good we have $v(g^z) > z + \thres$, then it goes to agent 1, and the remaining goods $g^z, g^y$ go to agent 2. This implies that $v(B_1) \geq z + \thres - \dm = z - \incr$, and $v(\xset{B}_2) \leq z + \left( \dm - \thres \right) = z + \incr$. Finally, from the normalization condition we have $\sum_{j \in [k]}p(g_j^x) = 1 - 2z - y \geq 1 - 3z$, and furthermore, $\sum_{j \in [k]}p(g_j^x) + z = p(A_1) < \pAoneUB$, which imply that $1 - 3z < \pAoneUB - z$, or equivalently, $z > \frac{3 + a}{(2 + a)(5 - a)}$. So, $B$ is an \praefx allocation with $a' \geq \frac{z - (3 + a)(1 - a)/((2 + a)(5 - a)(1 + a))}{z + (3 + a)(1 - a)/((2 + a)(5 - a)(1 + a))} > \frac{1 - (1-a)/(1+a)}{1 + (1-a)/(1+a)} = a$, for all $a \in (\varphi - 1, 1)$. 
    \end{enumerate}

    \item[\underline{Form 4:}] $A_2$ contains a $g^x$ good. There are the following cases.

    (i) $A_2 = \{ g^y, g^x \}$; then, $A_1$ must contain a $g^z$ good for the same reason as above. If it contains at least two $g^z$ goods or a $g^z$ and a $g^y$ good, then $p(A_1) > p(A_2)$, a contradiction. Therefore, it must be $A_1 = \{ g_2^z, g_1^x, g_2^x, \dots, g_k^x \}$ for some $k \geq 0$. The analysis of this case is omitted since it is similar to that of Form 2, with the only modification being that instead of a $g^y$ good in $A_2$ we have a $g^x$ good, where $p(g^x) \geq p(g_j^x)$ for all $j \in [k]$.
    
    (ii) $A_2 = \{ g^z, g^x \}$; then $A_1$ must contain a $g^y$ good for the same reason as above. If $A_1$ also contains a $g^z$ good, then $p(A_1) > p(A_2)$, a contradiction. Therefore, it must be $A_1 = \{ g_1^y, g_2^y, \dots, g_\ell^y, g_1^x, g_2^x, \dots, g_k^x \}$ for some $\ell \geq 1, k \geq 0$. If $\ell + k = 1$, then $k=0$ and the prediction has only three goods, and since $\dmax \leq \dmaxthree$ for all $a \in (\varphi - 1, 1)$, \cref{alg:3_goods} guarantees an \aefx allocation due to \cref{lem:3_goods}. If $\ell + k \geq 2$, we know that $A$ is an exact EFX allocation, so $p(A_1) \geq p(\xset{A}_2)$, which can be used to show that $p(g^x) = p(A_2) - p(\xset{A}_2) \geq p(A_2) - p(A_1) > \left( 1 - \pUB \right) - \pUB = 1 - 2 \cdot \pUB$. Let w.l.o.g., $p(g_1^x) \geq p(g_2^x) \geq \dots \geq p(g_k^x)$. There are two subcases.
    \begin{enumerate}
        \item [(I)] $p(g^x) > p(g_1^x)$; then $\ell \geq 2$, otherwise LPT would have allocated $g^x$ to $A_1$, since $y < z$. So in this case, $p(A_1) \geq 2y \geq 2 p(g^x)$, where the last inequality is by definition of $g^x$. 

        \item [(II)] There is a $j \in [k]$ such that $p(g^x) \leq p(g_{j}^x)$. So, we have $p(A_1) \geq y + p(j) \geq 2 p(g^x)$.
    \end{enumerate}
    
    In both subcases we have $p(A_1) \geq 2 p(g^x)$. Furthermore, notice that $2 \left( 1 - 2 \cdot \pUB \right) \geq \pUB $, since $\dm \leq \frac{2}{5} \frac{1-a}{1+a}$ for all $a \in (\varphi - 1, 1)$. Therefore, we have $ p(A_1) \geq 2 p(g^x) \geq 2 \left( 1 - 2 \cdot \pUB \right) \geq \pUB > p(A_1)$, a contradiction.

    (iii) $A_2 = \{ g_1^x, g_2^x \}$; then, $A_1$ must contain a $g^z$ and a $g^y$ good (by definition, there cannot be a $g^x$ without a $g^y$ and a $g^z$ good), therefore, $p(A_1) > p(A_2)$, a contradiction.
\end{enumerate}

Finally, notice that \cref{alg:id_2_ag} first reads the prediction, finds an exact EFX allocation $A$ using the LPT algorithm as a subroutine, which takes $O(T' \log T')$ time (see \cref{prop:LPT}), and categorizes $A$ to a particular Form in time $O(T')$ by finding the smallest valued bundle to determine $A_1$, and checking in which bundle each good is. Then, as each good arrives online, it decides its allocation by doing a constant number of basic operations, that is, updating the total value of each bundle and comparing the true value of the good with its predicted one.
\end{proof}

\begin{example}\label{ex:alg_better_than_LPT}
    As we saw in \cref{ex:LPT}, for prediction accuracy $0.945$, LPT can guarantee an $0.718$-EFX allocation on the true values. According to \cref{thm:with_predictions_2_ag_ternary_a-EFX_positive}, for the same accuracy, \cref{alg:id_2_ag} guarantees an $0.734$-EFX allocation on the true values. Conversely, the latter algorithm suffices to have accuracy $0.941$, in order to provide an $0.718$-EFX allocation.
\end{example}

\subsubsection{Results on 2-value predictions}

One can notice that, even though \cref{alg:id_2_ag} works for two agents with any additive, identical valuation, in essence, it only considers the three highest levels of value inside the prediction vector $p$. This means that, even if we were restricted to have a 3-value predictor, the proof of \cref{thm:with_predictions_2_ag_ternary_a-EFX_positive} would not give any better bounds. However, in the case where we have a 2-value predictor, the same proof yields improved results. It is important to note that, even though the prediction is a 2-value function, the true valuation is not restricted at all.

\begin{corollary}[Corollary of \cref{thm:with_predictions_2_ag_ternary_a-EFX_positive}]\label{cor:with_predictions_2_ag_binary_a-EFX_positive}
    Suppose we have $2$ agents with additive, identical, normalized valuations, with a provided 2-value prediction of accuracy $\eta \geq 1 - \dmaxbin$ for some given $a \in \left( \varphi - 1, 1 \right]$, that is, the error between the prediction and the true valuation is $1 - \eta \leq \dmaxbin$. Then, \cref{alg:id_2_ag} outputs an \aefx allocation, and performs a constant number of basic operations per time-step.
\end{corollary}

\begin{proof}
    Observe that when the prediction is a 2-value function, the proof of \cref{thm:with_predictions_2_ag_ternary_a-EFX_positive} goes through for $\dm = \dmaxbin$. That is because \cref{prop:p(A_1)_large_enough} holds for any $\dm \leq \frac{1}{2} \cdot \frac{1-a}{1+a}$, \cref{prop:bound_card_A2} holds for any $\dm \leq \dmaxbin$, and, since there are no $g^x$ goods, $A$ can have either Form 1, Form 2, or Form 3, with $k=0$. The latter two forms involve only three predicted goods, so \cref{lem:3_goods} applies. Finally, Form 1 goes through for $\dm \leq \dmaxformone$. It turns out that the minimum of all the above bounds is $\dmaxbin$.
\end{proof}

It is natural that the accuracy upper and lower bounds improve in this case. However, notice that the latter improves disproportionately, indicating that for 2-value predictions, most likely there is an algorithm with better guarantees than \cref{alg:id_2_ag}.

\begin{theorem}\label{thm:with_predictions_2_ag_binary_a-EFX_inapprox_bound}
    Suppose we have $2$ agents with additive, identical, normalized valuations, with a provided 2-value prediction of accuracy $\eta < 1 - \frac{1-a}{2}$ for some given $a \in \left( \sqrt{3}-1, 1 \right]$, that is, the error between the prediction and the true valuation is $1 - \eta > \frac{1-a}{2}$. Then, there is no algorithm that guarantees an \aefx allocation, even when $T'=T=4$.
\end{theorem}

\begin{algorithm}[H]
\caption{Two agents with additive, identical valuations: Computing an \aefx for $a \in (\varphi - 1, 1]$ when prediction accuracy is at least $1 - \dmax$}\label{alg:id_2_ag}  
	\begin{algorithmic}[1]
		\REQUIRE{A prediction vector $(p(g_t))_{t \in [T']}$.}
		\ENSURE{An \aefx allocation $B$.}
		
		\medskip

        \STATE{$(A_1, A_2) \gets (\emptyset, \emptyset)$}
        \STATE{$(B_1, B_2) \gets (\emptyset, \emptyset)$}
        \STATE{$\dm \gets \dmax$}
        \STATE{$z \gets$ \text{ highest value in $p$}, \quad $y \gets$ \text{ $2^{nd}$ highest value in $p$}}
        \STATE{$g^z \in \{ g ~|~ p(g)=z \}$, $g^y \in \{ g ~|~ p(g)=y \}$, $g^x \in \{ g ~|~ p(g)<y \}$}
        
        \medskip
		
		\IF{$T' \leq 3$} 
                \STATE{$(B_1, B_2) \gets $\text{ Output of \cref{alg:3_goods}}}
                    \COMMENT{\cref{lem:3_goods} applies}
                \STATE{\textbf{go to} line \ref{alg:line_return}}
                    
		\ENDIF
        
        \medskip

        \STATE{$(A_1, A_2) \gets $\text{ Output of} \cref{alg:LPT} with input $p$ } 
            \COMMENT{$A$ is exact EFX allocation based on $p$}

        \medskip 
        
		\IF{$p(A_1) \geq \pAoneUB$ \OR $a = 1$} 
			\STATE{$(B_1, B_2) \gets (A_1, A_2)$}
                \COMMENT{\cref{prop:p(A_1)_large_enough} applies, or $a = 1$ (therefore prediction accuracy is $1$)}
            \FOR{$t \in \{ T'+1, T'+2, \dots, T \}$}
                \STATE{$B_{1} \gets B_{1} \cup \{ g_t \}$} 
                \COMMENT{All goods $g_t$ with $t > T'$ (if any) get arbitrarily allocated to agent 1}
            \ENDFOR
        \medskip
        
		\ELSE
                \IF{Form 1: $A_2 = \{ g_1^z, g_2^z \}$, $A_1 = \{ g_3^z, g_1^y, g_2^y, \dots, g_\ell^y, g_1^x, g_2^x, \dots, g_k^x \}$ for $\ell + k \geq 1$}
				\IF{$g_t = g^z$ with $v(g_t) \leq z + \frac{\dm}{2}$ \AND $|A_2| \leq 1$} 
					\STATE{$B_2 \gets B_2 \cup \{ g_t \}$}
                \ELSE
                    \STATE{$B_1 \gets B_1 \cup \{ g_t \}$}
				\ENDIF
			\ENDIF
            
            \medskip
            
			\IF{Form 2 or Form 4: $A_2 = \{ g_1^y, g^* \}$, where $g^* \in \{ g_2^y, g^x \}$, $A_1 = \{ g^z, g_1^x, g_2^x, \dots, g_k^x \}$, for $k \geq 1$}
                    \IF{$g_t \in \{g_1^y, g^*, g^z\}$ with ( $ v(g_t) \leq y + \frac{\dm}{2}$ \AND $|A_2| \leq 1 $ )  \OR $ | \{g_1^y, g^*, g^z\} \cap A_1 | \geq 1$} 
						\STATE{$B_2 \gets B_2 \cup \{ g_t \}$}
                    \ELSE
                        \STATE{$B_1 \gets B_1 \cup \{ g_t \}$}
                    \ENDIF
			\ENDIF 
            
            \medskip
            
			\IF{Form 3: $A_2 = \{ g_1^z, g_1^y \}$, $A_1 = \{ g_2^z, g_1^x, g_2^x, \dots, g_k^x \}$ for $k \geq 1$}
				\IF{in $p$, good $g^y$ does not arrive after both $g^z$ goods}
					\IF{$g_t \in \{g_1^y, g_1^z, g_2^z\}$ with $v(g_t) \leq z + \frac{\dm}{2}$ \AND $|A_2| \leq 1$} 
						\STATE{$B_2 \gets B_2 \cup \{ g_t \}$}
                    \ELSE
                        \STATE{$B_1 \gets B_1 \cup \{ g_t \}$}
					\ENDIF
				\ELSE
                    \IF{$g_t \in \{g_1^y, g_1^z, g_2^z\}$ with ( $v(g_t) \leq z + \thres$ \AND $|A_2| \leq 1$ ) \OR $| \{ g_1^y, g_1^z, g_2^z\} \cap A_1 | \geq 1$}
                        \STATE{$B_2 \gets B_2 \cup \{ g_t \}$}
                    \ELSE
                        \STATE{$B_1 \gets B_1 \cup \{ g_t \}$}
                    \ENDIF
				\ENDIF
            \ENDIF
		\ENDIF
        
        \medskip

        \RETURN{$(B_1, B_2)$}\label{alg:line_return}
	\end{algorithmic}
\end{algorithm}

\begin{proof}
    For the sake of contradiction, suppose there is an algorithm that for some $a \in (\sqrt{3}-1, 1]$ and accuracy $\eta < 1 - \frac{1-a}{2}$, guarantees an \aefx allocation. The adversary fixes a rational $\varepsilon \in \left( \frac{1-a}{2}, \frac{a}{2(2+a)} \right)$, and notice that this interval is non-empty only for $a \in (\sqrt{3}-, 1]$. Let the number of rounds be $T := 4$, which is known by the algorithm (since it is given the prediction itself). 
        
    The adversary gives to the algorithm the prediction: $p(g_1) = p(g_2) = \varepsilon$, $p(g_3) = p(g_4) = \frac{1}{2} - \varepsilon$. 
    The adversary reveals the true values $v(g_1) = v(g_2) = \varepsilon$.    
    There are two cases:

    (i) $g_1, g_2$ get allocated to the same agent, w.l.o.g. agent 1. Then, the adversary reveals the true values $v(g_3) = v(g_4) = \frac{1}{2} - \varepsilon$. if both $g_3, g_4$ get allocated to agent 1, agent 2 is $0$-EFX envious towards him, therefore this is not an \aefx allocation as requested. If both $g_3, g_4$ get allocated to agent 2, then this is an \praefx allocation with $a' = \frac{2 \varepsilon}{1/2 - \varepsilon} < a$, where the inequality comes from the fact that $\varepsilon < \frac{a}{2(2+a)}$. If, w.l.o.g., $g_3$ goes to agent 1 and $g_4$ goes to agent 2, then this is an \praefx allocation with $a' = \frac{1/2 - \varepsilon}{1/2} < a$, where the inequality is due to the fact that $\varepsilon > \frac{1-a}{2}$.

    (ii) W.l.o.g., $g_1$ gets allocated to agent 1, and $g_2$ gets allocated to agent 2. Then the adversary reveals the true values $v(g_3) = \frac{1}{2} - 2 \varepsilon$, $v(g_4) = \frac{1}{2}$. If $g_3, g_4$ get allocated to the same agent, then this is an \praefx allocation with $a' = \frac{\varepsilon}{1 - 2 \varepsilon} \leq \frac{2 \varepsilon}{1/2 - \varepsilon} < a$, where the last inequality comes from the fact that $\varepsilon < \frac{a}{2(2+a)}$. If, w.l.o.g., $g_3$ goes to agent 1 and $g_4$ goes to agent 2, then this is an \praefx allocation with $a' = \frac{1/2 - \varepsilon}{1/2} < a$, where the inequality is due to the fact that $\varepsilon > \frac{1-a}{2}$.

    Therefore, it is not possible for the algorithm to output an \aefx allocation. Also, notice that the TV distance between the prediction and the adversary is $\varepsilon$ which can take any value in $\left( \frac{1-a}{2}, \frac{a}{2(2+a)} \right)$.
\end{proof}

For more than 2 agents, we show the following lower bound on the accuracy.

\begin{theorem}\label{thm:with_predictions_n_ag_binary_a-EFX_inapprox_comb}
    Suppose we have $n \geq 3$ agents with additive, identical, normalized valuations, with a provided 2-value prediction of accuracy $\eta < 1 - \frac{2(1 - a^2)}{4 + (2n-3)a}$ for some given $a \in \left( 0, 1 \right]$, that is, the error between the prediction and the true valuation is $1 - \eta > \frac{2(1 - a^2)}{4 + (2n-3)a}$. Then, there is no algorithm that guarantees an \aefx allocation, even when $T' = T = 2n-1$.
\end{theorem}

\begin{proof}
     The proof is almost identical to that of \cref{lem:with_predictions_n_ag_ternary_a-EFX_inapprox}, and it is omitted to avoid repetition. Consider an $a \in (0,1]$ and an algorithm with prediction accuracy $\eta < 1 - \frac{2(1-a^2)}{4 + (2n-3)a}$. The number of goods and the case analysis are the same as those of the aforementioned result. What changes is the prediction given by the adversary, as well as the true values revealed in each case. In particular, the adversary fixes rational values $k \in \left( 0, \frac{a}{4 + (2n-3)a} \right)$ and $\varepsilon \in \left( \frac{2(1-a^2)}{4 + (2n-3)a}, \frac{2}{4 + (2n-3)a} \right]$ and gives the prediction $p(g_1) = p(g_2) = \dots = p(g_{2n-3}) = k$, $p(g_{2n-2}) = p(g_{2n-1}) = \frac{1-(2n-3)k}{2}$. The adversary reveals the true value of the first $2n-3$ goods (called ``$k$-goods''), namely, $v(g_1) = v(g_2) = \dots = v(g_{2n-3}) = k$. After time-step $t=2n-3$, in case (i) it reveals true values $v(g_{2n-2}) = v(g_{2n-1}) = \frac{1-(2n-3)k}{2}$, and in cases (ii) and (iii) it reveals $v(g_{2n-2}) = \frac{1-(2n-3)k}{2} - \varepsilon$, $v(g_{2n-1}) = \frac{1-(2n-3)k}{2} + \varepsilon$. 

    For the specified domains of $k$ and $\varepsilon$, using the same analysis as that of \cref{lem:with_predictions_n_ag_ternary_a-EFX_inapprox}, we can show that no matter where the algorithm allocates the goods, it can never produce an \aefx allocation. Notice here that the error between the prediction and the true valuation in all the above cases is $\varepsilon$, which can take any value in $\left( \frac{2(1-a^2)}{4 + (2n-3)a}, \frac{2}{4 + (2n-3)a} \right]$. 
\end{proof}

\section{Conclusions and open problems}

We study the problem of computing (approximate) EFX allocations when goods of unknown value arrive over time to a given set of agents. We show inapproximability results even in cases where exact EFX allocations in the offline setting are easy to compute. Our negative results are unconditional, caused by the incomplete information of the setting. To break this information-related barrier, we consider algorithms enhanced with predictions and derive lower and upper bounds (i.e., necessary and sufficient values) of the accuracy levels that guarantee an $a$-EFX allocation for any given $a \in [0,1]$. We then provide an efficient algorithm that carefully uses the predictions and the observed values to bridge the bound-gap.

There are several open questions stemming from this work. First, for the simplest case of 2 agents with additive, identical valuations, the gap should be closed either by finding a better upper bound than that of \cref{thm:with_predictions_2_ag_ternary_a-EFX_positive}, or improving the lower bound of \cref{thm:with_predictions_comb-beyond-phi-EFX_inapprox_bound}. Also, future work should aim to close the corresponding gap for more than 2 agents, and do the same for the general additive valuations case. Additionally, the family of algorithms that do not have access to the prediction accuracy deserve further investigation. Finally, it would be interesting to study this model when, instead of goods, chores or a mixture of goods and chores are considered.

\bigskip
\subsection*{Acknowledgements}
This work has been partially supported by project MIS 5154714 of the National Recovery and Resilience Plan Greece 2.0 funded by the European Union under the NextGenerationEU Program.

\bibliographystyle{alpha}
\bibliography{bibliography}

\end{document}